\documentclass[a4paper]{article}
\pdfoutput=1
\usepackage{macros}
\newcommand{\IEEEhspace}[1]{}
\newcommand{\appref}[1]{App.~\ref{#1}}
\usepackage{amsfonts,amstext,amsmath,amssymb,stmaryrd,calc,amsthm,mathtools,bold-extra}
\usepackage{gb4e}
\theoremstyle{plain}

\newtheorem{theorem}{Theorem}
\newtheorem{proposition}{Proposition}
\newtheorem{corollary}{Corollary}
\theoremstyle{definition}
\newtheorem*{pfmc}{Parse Forest Model-Checking Problem}
\newtheorem*{pp}{{\boldmath\PDL\ }Recognition Problem}

\newtheorem{example}{Example}
\theoremstyle{remark}

\newtheorem{claim}[theorem]{Claim}
\usepackage{subfig,tikz}
\usetikzlibrary{positioning}
\usetikzlibrary{arrows}
\usetikzlibrary{automata}
\usetikzlibrary{trees}
\makeatletter
\def\@listi{\leftmargin\leftmargini
               \topsep 3\p@ \@plus\p@ \@minus\p@
               \parsep 2\p@ \@plus\p@ \@minus\p@
               \itemsep \parsep}
\makeatother
\usepackage{natbib}\renewcommand{\cite}{\citep}
\usepackage{url,hyperref}

\title{Model Checking Parse Trees}
\author{Anudhyan Boral$^{\text{1}}$\and Sylvain Schmitz$^{\text{2}}$}
\date{$^{\text{1}}$CMI, Chennai, India\\$^{\text{2}}$LSV, ENS Cachan
  \& CNRS, France}
\begin{document}
\maketitle
\begin{abstract}
  Parse trees are fundamental syntactic structures in both
  computational linguistics and compilers construction.  We argue
  in this paper that, in both fields, there are good incentives for
  model-checking sets of parse trees for some word according to a
  context-free grammar.  We put forward the adequacy of propositional
  dynamic logic (PDL) on trees in these applications, and study as a
  sanity check the complexity of the corresponding model-checking
  problem: although complete for exponential time in the general case,
  we find natural restrictions on grammars for our applications and
  establish complexities ranging from nondeterministic polynomial time
  to polynomial space in the relevant cases.
\end{abstract}

\section{Introduction}\label{sec-intro}

The parse trees of a sequence $w$ are employed extensively in
computational linguistics, where they represent constituent analyses
of the natural language sentence $w$, and in compilers, where they
provide the syntactic structure of the input program $w$.  They are
produced by a parsing algorithm on the basis of a grammar, for
instance a context-free one, that typically required quite a bit of
ingenuity in its conception: in each of the two communities, a
subfield of \emph{grammar engineering} has arisen
\citep[e.g.][]{grammarware,tulipa}, dedicated to the principled
development of grammars.  Part of the difficulty in this task stems
from the limited expressiveness of the formalism: for instance, in a
context-free grammar, a production $A\to B\,C$ only states that an
`$A$' node might have two children labeled `$B$' and `$C$' in this
order.  This means that orthogonal considerations, possibly involving
distantly related tree nodes, have to be enforced manually, little by
little in the local rules of the grammar, blowing up its size and
quickly rendering it unmaintainable.

\emph{Model theoretic syntax} provides an alternative approach to
syntax specifications: rather than an `enumerative' formalism that
generates the desired parse trees, one can describe them as models of
a logical formula.  This point of view leads to interesting
consequences for syntactic theory~\citep{pullum01}, but here we are
mostly interested in the conciseness and ease of manipulation of
formul\ae.  Several logical formalisms over trees have been employed to
this end, notably first-order (FO) or monadic second-order (MSO)
logics~\citep{rogers98}, and \emph{propositional dynamic logic on
  trees} (\PDL)~\citep{kracht95,palm99,pdltree}.  In practice however,
model-theoretic approaches suffer from a prohibitively high
complexity, as the known recognizing algorithms essentially amount to a
\emph{satisfiability} check~\citep{cornell00}: given a formula $\varphi$ and
a sequence $w$ to parse, build a formula $\varphi_w$ that recognizes
all the trees that yield $w$, and check $\varphi\wedge\varphi_w$ for
satisfiability.  As satisfiability is in general non-elementary for FO
and MSO formul\ae\ and \textsc{ExpTime}-complete for \PDL\ ones, this
seems like a serious impediment to a larger adoption of
model-theoretic techniques.
\medskip

In this paper, we introduce the \emph{model-checking} problem for parse
trees.  Formally, given a sequence $w$, a grammar $G$,
and a formula $\varphi$, we ask whether all the parse trees of $w$
according to $G$ verify $\varphi$.  It turns out that checking sets of
parse trees of a given $w$, i.e.\ \emph{parse forests}, can be
easier than other classes of tree languages, as could be defined by
tree automata, document type definitions, etc.

This parse forest model-checking problem (PFMC) allows for a `mixed'
approach, where a context-free grammar is employed for a cursory
syntax specification, alongside a logical formula for the fine-tuning.
Because the logical languages we consider are closed under negation,
under this viewpoint, the PFMC problem also answers the recognition
problem for the `conjunction' of the grammar $G$ and the formula
$\varphi$: is there a parse of $w$ according to both $G$ and
$\varphi$?  As $\varphi$ might describe a non-local tree language,
there is a slight expressive gain to be found in such conjunctions,
but our interest lies more in the concision and clarity brought by
refining a grammar with a $\PDL$ formula: it allows to capture
long-distance dependencies that would often require a cumbersome and
error-prone hard-wiring in the grammar, at the expense of an explosion
of the number of nonterminal symbols.

We detail in \autoref{sec-mcpf} two applications of the PFMC; for
these, we found it convenient to use \PDL\ as the logical language for
properties:
\begin{enumerate}
\item In computational linguistics, we advocate a mixed
  approach for model-theoretic syntax, with syntactic structures
  described by the conjunction of a grammar for localized
  specification together with a \PDL\ constraint capturing
  long-distance syntactic phenomena;
\item In compilers construction, 
  \PDL\ formul\ae\ provide a compelling means for parser
  disambiguation~\citep{amb/thorup,amb/klint,Kats2010} by allowing to
  express formally the informal disambiguation rules usually provided
  with grammars for programming languages.
\end{enumerate}
We discuss the appropriateness of our formalizations in some depth,
which allows us to motivate (1)~practically relevant restrictions on
the grammar $G$, and (2)~considering the full logic \PDL\ rather than
some weaker fragments.  We consider the two formalizations proposed in
\autoref{sec-mcpf} as the initial steps of a larger research programme
on the model-checking problem in syntax; we point for instance to
several interesting open issues with the choice of finite labeled
ordered trees as syntactic structures and \PDL\ as logical formalism.

As a first usability check, we investigate the computational
complexity of the PFMC and map the resulting complexity landscape for
the problem in \autoref{sec-cmpl}.  Although the general case is
\textsc{ExpTime}-complete like the \PDL\ satisfiability problem, our
restrictions on grammars lead to more affordable complexities
\begin{enumerate}
\item from \textsc{NPTime}-complete for our linguistic applications,
  where we can assume the grammar to be both $\varepsilon$-free and
  acyclic, 
\item to \textsc{PSpace}-complete for our applications in ambiguity
  filtering, where we can only assume the grammar to be acyclic
\end{enumerate}
(see \autoref{fig-cmpl} for a summary).  Our study also unearthed a
somewhat surprising corollary for model-theoretic syntax
(Cor.~\ref{coro-mts}):
\begin{enumerate}\setcounter{enumi}{2}
\item the \emph{recognition problem} for \PDL, i.e.\ whether there
  exists a tree model with the input word as yield, is
  \textsc{PSpace}-complete if empty labels are forbidden---the best
  algorithms for this were only known to operate in exponential
  time~\citep{cornell00,palm01}.
\end{enumerate}

Interestingly, the PFMC is closely related to a prominent algorithmic
problem studied by the XML community: there the formula $\varphi$ is a
\emph{Core XPath} one---which is equivalent to a restricted fragment
\PDLcr\ of \PDL---and the tree language $L$ is generated by a
\emph{document type definition} (DTD), and the problem is accordingly
referred to as `satisfiability in presence of a DTD'.
\Citet{benedikt08} comprehensively investigate this topic, and in some
restricted cases the problem becomes
tractable~\citep{montazerian07,ishihara09}.  Our applications in
computational linguistics and compilation lead however to a different
setting, where $L$ comes from a class of tree languages smaller than
that of DTDs---and our grammar restrictions have no natural
counterpoints in the XML literature---, but where $\varphi$ requires
more expressive power than that of \PDLcr.  Nevertheless, we will
reuse several proof techniques developed in the XML setting and adapt
them to ours, Prop.~\ref{prop-efpe} being a prime example: it relies
on an extension of the results of \citet{benedikt08} to the full logic
\PDL\ (see Prop.~\ref{prop-dtd}) and on an encoding of a restricted
class of parse forests into non-recursive DTDs.

\section{Propositional Dynamic Logic on Trees}\label{sec-pdl}
\emph{Propositional dynamic logic} (PDL, see \citep{fischer79}) is a
modal logic where ``programs''---in the form of regular expressions
over the relations in a frame---are used as modal operators.  
Originally motivated by applications in computational
linguistics~\citep{kracht95,palm99,pdltree}, \emph{propositional dynamic
logic on trees} (\PDL) has also been extensively studied in the XML
community~\citep{marx05,benedikt08,tencate10}, where it is better
known as \emph{Regular XPath}.  It features two
relations: the \emph{child} relation $\downarrow$ between a parent
node and any of its immediate children, and the \emph{right-sibling}
relation $\rightarrow$ between a node and its immediate right sibling.

\subsection{Syntax and Semantics}
Formally, a \PDL\ formula $\varphi$ is defined by the abstract syntax
{\begin{align*}
  \varphi &::= p\mid\top\mid\neg\varphi\mid\varphi\wedge\varphi\mid\tup{\pi}\varphi\tag{node formul\ae}\\
  \pi &::={\downarrow}\mid{\rightarrow}\mid\pi\mathbin{;}\pi\mid\pi+\pi\mid\pi^\ast\mid\pi^{-1}\mid\varphi?\tag{path
    formul\ae}
\end{align*}}%
where $p$ is an atomic proposition ranging over some countable set
$\#{AP}$---because we only deal with satisfiability questions, we can
actually assume $\#{AP}$ to be finite.  We enrich this syntax as usual
by defining box modalities as duals
$[\pi]\varphi\eqdef\neg\tup{\pi}\neg\varphi$ of the diamond ones, inverses
to the atomic path formul\ae\ as ${\uparrow}\eqdef{\downarrow^{-1}}$
and ${\leftarrow}\eqdef{\rightarrow^{-1}}$, and boolean connectives
${\bot}\eqdef\neg\top$,
$\varphi_1\vee\varphi_2\eqdef\neg(\neg\varphi_1\wedge\neg\varphi_2)$,
$\varphi_1\imply\varphi_2\eqdef\neg\varphi_1\vee\varphi_2$, and
$\varphi_1\equiv\varphi_2\eqdef(\varphi_1\imply\varphi_2)\wedge(\varphi_2\imply\varphi_1)$.

Formul\ae\ are interpreted over \emph{finite ordered trees} $t$ with
nodes labeled by propositions in $\#{AP}$.  Such a tree $t$ is a partial
function from the set $\+N^\ast$ of finite sequences of natural
numbers, i.e.\ the set of tree nodes, to $\#{AP}$, %
s.t.\ its \emph{domain}
$\dom t$ is (1)~\emph{finite}, (2)~\emph{prefix closed}, i.e.\ $uv$ in
$\dom t$ for some $u,v$ in $\+N^\ast$ implies that $u$ is also in
$\dom t$, and (3)~\emph{predecessor closed}, i.e.\ if $ui$ is in $\dom
t$ for some $u$ in $\+N^\ast$ and $i$ in $\+N$, then $uj$ is also in
$\dom t$ for all $j<i$ in $\+N$.  Such a tree can be seen as a
structure $\mathfrak{M}_t=\tup{\dom t,\downarrow_t,\rightarrow_t,t}$
with
{\begin{align*}
 {\downarrow_t}&\eqdef\{(u,ui)\mid ui\in\dom
t\} \\{\rightarrow_t}&\eqdef\{(ui,u(i+1))\mid u(i+1)\in\dom t\}\;.
\end{align*}}%
We define the \emph{interpretations} of \PDL\ formul\ae\ over $t$
inductively by
{\begin{align*}
    \sem{p}_t &\eqdef\{u\in\dom t\mid p=t(u)\}\!\!
  & \sem{\top}_t&\eqdef\dom t
  \\\sem{\varphi_1\wedge\varphi_2}_t&\eqdef\sem{\varphi_1}_t\cap\sem{\varphi_2}_t
  & \sem{\neg\varphi}_t&\eqdef\dom t\setminus\sem{\varphi}_t
  \\\sem{\varphi?}_t&\eqdef\{(u,u)\mid u\in\sem{\varphi}_t\}
  &\!\!\!\!\!\sem{\tup{\pi}\varphi}_t&\eqdef\sem{\pi}_t^{-1}(\sem{\varphi}_t)
  \\\sem{\downarrow}_t&\eqdef{\downarrow_t}
  &\sem{\rightarrow}_t&\eqdef{\rightarrow_t}
  \\\sem{\pi_1\mathbin{;}\pi_2}_t&\eqdef\sem{\pi_1}_t\fatsemi\sem{\pi_2}_t
  &\sem{\pi^\ast}_t&\eqdef\sem{\pi}_t^\ast
  \\\sem{\pi_1+\pi_2}_t&\eqdef\sem{\pi_1}_t\cup\sem{\pi_2}_t
  & \sem{\pi^{-1}}_t&\eqdef\sem{\pi}_t^{-1}
\end{align*}}%
Observe that these are sets of nodes included in $\dom t$ in the case
of node formul\ae, but binary relations included in $\dom t\times\dom
t$ in the case of path formul\ae; thus $\sem{\pi}_t^\ast$ denotes the
reflexive transitive closure of $\sem{\pi}_t$ and $\sem{\pi}_t^{-1}$
its inverse, while $\sem{\pi_1}_t\fatsemi\sem{\pi_2}_t$ denotes the
composition of the two relations $\sem{\pi_1}_t$ and $\sem{\pi_2}_t$.
A node $u$ in $\dom t$ \emph{satisfies} $\varphi$, noted
$t,u\models\varphi$, if $u$ is in $\sem{\varphi}_t$.  A tree $t$
satisfies $\varphi$, noted $t\models\varphi$, if its root
$\varepsilon$ satisfies $\varphi$; we let $\sem\varphi\eqdef\{t\mid
t\models\varphi\}$ be the set of models of $\varphi$.

\begin{example}[Basic Navigation]\label{ex-nav}
  Several simple formul\ae\ helping navigation can be defined:
  $\#{root}\eqdef\neg\tup\uparrow\top$ holds only at the root,
  $\#{leaf}\eqdef\neg\tup\downarrow\top$ only at a leaf node,
  $\#{first}\eqdef\neg\tup\leftarrow\top$ at a leftmost one, and
  $\#{last}\eqdef\neg\tup\rightarrow\top$ at a rightmost one.

  We can also define the \emph{first-child} relation
  ${\swarrow}\eqdef{\downarrow};\#{first}?$, and conversely express
  the child relation as
  ${\downarrow}\equiv{\swarrow};\rightarrow^\ast$: this shows that we
  could work on binary tree models instead of the unranked ones we
  used in our definitions.
\end{example}

\begin{example}[Parse Trees \citep{blackburn93}]\label{ex-cfg}
  Recall that a \emph{context-free grammar} (CFG) is a tuple
  $G=\tup{N,\Sigma,P,S}$ composed of a finite \emph{nonterminal}
  alphabet $N$, a finite \emph{terminal} alphabet $\Sigma$ disjoint
  from $N$ and forming a \emph{vocabulary} $V\eqdef
  N\uplus\Sigma$, a finite set of \emph{productions} $P\subseteq
  N\times V^\ast$, and an \emph{axiom} $S\in N$.  We denote the empty
  sequence by $\varepsilon$ and write
  $\Sigma'\eqdef\Sigma\uplus\{\varepsilon\}$ and
  $V'\eqdef V\uplus\{\varepsilon\}$.

  Given a context-free grammar $G$, its set of parse trees forms a
  local tree language, which can be expressed as $\sem{\varphi_G}$ for
  a \PDL\ formula $\varphi_G$ with $V'$ as set of atomic
  propositions.  First define a path formula $\pi_\alpha$ that defines
  a sequence of sibling nodes labeled by $\alpha$ in $V^\ast$:
  {\begin{align*}
    \pi_\alpha&\eqdef\begin{cases}
      X?;\rightarrow;\pi_{\alpha'}&\text{if }\alpha=X\alpha',X\in
      V,\alpha'\neq\varepsilon\:,\\
      X?;\#{last}?&\text{if }\alpha=X\in V\:,\\
      \varepsilon?;\#{last}?&\text{otherwise, i.e.\ if }\alpha=\varepsilon\:.
      \end{cases}\\
    \varphi_G&\eqdef S\tag{the root is labeled by $S$}\\
    &\wedge[\downarrow^\ast]\big(%
    \#{leaf}\equiv\bigvee_{a\in\Sigma'}a\tag{leaves
      are terminals and internal nodes nonterminals}\\
    &\phantom{\wedge[\downarrow^\ast]\big(}\wedge\bigwedge_{A\in
      V}A\imply\bigvee_{A\to \alpha}\tup{\swarrow;\pi_\alpha}\top\big)\,.\tag{productions are enforced}
  \end{align*}}%
\end{example}

\subsection{The Conditional Fragment}

We will consider in this paper several fragments of \PDL, most
importantly the \emph{conditional path} fragment \PDLcp\
\citep{palm99,marx05}, with a restricted syntax on path formul\ae\
{\begin{align*}
  \pi&::=\alpha\mid\pi\mathbin{;}\pi\mid\pi+\pi\mid\varphi?\mid(\alpha;\varphi?)^\ast\tag{conditional paths}\\
  \alpha&::={\leftarrow}\mid{\rightarrow}\mid{\uparrow}\mid{\downarrow}\;.\tag{atomic
    paths}
\end{align*}}%
This fragment is of particular relevance, because it extends the
\emph{core language} \PDLcr\ \citep{blackburn95,gottlob02} (which
features $\alpha^\ast$ instead of $(\alpha;\varphi?)^\ast$) and captures
exactly first-order logic over finite ordered trees with the two
relations $\rightarrow^+$ and $\downarrow^+$ \citep{marx05}.

\begin{example}[Depth-First Traversal]\label{ex-gs}
  Observe that the formul\ae\ in examples~\ref{ex-nav}
  and~\ref{ex-cfg} are actually in \PDLcr.  The \emph{depth-first
    traversal} relation ${\prec}\eqdef
  (\#{last}?;\uparrow)^\ast;\rightarrow;(\downarrow;\#{first}?)^\ast$
  is an example of a path that is not definable in \PDLcr---this can
  be checked for instance using an Ehrenfeucht Fra\"iss\'e argument.
\end{example}
More generally, \PDLcp\ allows to express relations akin to
LTL's \emph{until} and \emph{since} modalities; see
e.g.\ \citep{LS-jal10}.
We denote by $\PDL[\downarrow]$ (resp.\ $\PDLcp[\downarrow]$,
$\PDLcr[\downarrow]$) the fragments with only downward navigation,
i.e.\ without the $\rightarrow$, $\leftarrow$, and $\uparrow$ atomic
paths.

\section{Model-Checking Parse Trees}\label{sec-mcpf}
Many problems arising naturally with \PDL\ are decidable, notably the
\begin{description}%
\item[model-checking]\IEEEhspace{3.5em} problem: given a tree $t$ and a formula
  $\varphi$, does $t\models\varphi$?  This is known to be in
  \textsc{PTime} even for larger fragments of PDL~\citep{lange06}.
\item[satisfiability]\IEEEhspace{2.1em} problem: given a formula $\varphi$, does there
  exist a tree $t$ s.t.\ $t\models\varphi$?  This is known to be
  \textsc{ExpTime}-complete~\citep{pdltree}.
\end{description}
In the context of XML processing and XPath, an intermediate question
between model-checking and satisfiability also arises:
\begin{description}
\item[satisfiability in presence of a tree
language\textrm{:}]\IEEEhspace{15em} given a formula
  $\varphi$ and a regular tree language $L$, does there exist a tree
  $t\in L$ s.t.\ $t\models\varphi$?
\end{description}
Due to its initial XML motivation, the basic case for this problem is
that of a $\PDLcr[\downarrow]$ formula (a downward Core XPath query)
and of a local tree language (described by a DTD), but many variants
exist~\citep{benedikt08,montazerian07,ishihara09}---in particular one
where the tree language is the language of infinite trees of a two-way
alternating parity tree automaton, which is used by \citet{goller09}
to prove that the satisfiability problem for PDL with intersection and
converse is 2\textsc{ExpTime}-complete.

Our own flavour is motivated by applications in computational
linguistics and programming languages, where the tree language is the
set of parse trees of a word $w$ in $\Sigma^\ast$ according to a CFG
$G=\tup{N,\Sigma,P,S}$ verifying $V'\eqdef\Sigma\uplus
N\uplus\{\varepsilon\}=\#{AP}$.
More precisely, following a well-known construction of
\citet{barhillel61}, if $w=a_1\cdots a_n$ is a word of length $n$, the
set of parse trees or \emph{parse forest} of a CFG
$G$ for $w$, written $L_{G,w}$, is the regular tree language
recognized by a tree automaton $\mathcal{A}_{G,w}$ with state set
{\begin{align*}
  Q_{G,w}&\eqdef\!\{(i,X,j)\mid 0\leq i\leq j\leq n, X\in V'\}\;,
\intertext{{\normalsize alphabet $V'$, initial state $(0,S,n)$, and rules}}
  \delta_{G,w}&\eqdef\{(i_0,A,i_m)\to
    A((i_0,X_1,i_1)\cdots(i_{m-1},X_m,i_m))\\
    &\hspace{2em}\mid A\to X_1\cdots X_m\in P%
    \wedge 0\leq i_0\leq\cdots\leq i_m\leq n\}\\
    &\,\cup\,\{(i,a_{i+1},i+1)\to a_{i+1}()\mid 0\leq i<n\}
    \\ &
    \,\cup\,\{(i,\varepsilon,i)\to\varepsilon()\mid 0\leq i\leq n\}\;.
  \end{align*}}%
Intuitively, a state $(i,X,j)$ of this automaton recognizes the set of
trees derivable in $G$ from the symbol $X$ and spanning the factor
$a_{i+1}\cdots a_j$ of $w$.  This automaton is in general not
\emph{trim}, in that many of its states and rules are never employed
in any accepting configuration, but it can be trimmed in linear
time if required.

\begin{pfmc}[PFMC]\hfill\nopagebreak
\begin{description}
\item[input] a context free grammar $G$, a word $w$, and a \PDL\
  formula $\varphi$,
\item[question]\IEEEhspace{1em} does there exists $t\in L_{G,w}$ s.t.\
  $t\models\varphi$?\qed
\end{description}
\end{pfmc}

Note that the automaton $\?A_{G,w}$ has size $O(|G|\cdot |w|^{m+1})$
if $m$ is the maximal length of a production rightpart in $G$; since
the grammar can be put in quadratic form (corresponding to the
binarization we would also perform on the formula), this typically
results in size $O(|G|\cdot |w|^{3})$.  Therefore, although a tree
automaton for the tree language is not part of the input, it can
nevertheless be constructed in logarithmic space.  
The originality of the problem stems from considering parse forests,
which form a rather restricted class of tree languages.  

In \autoref{sec-cmpl}, we will investigate the complexity of this
problem, and focus on the influence of the acyclicity and
$\varepsilon$-freeness of $G$: Define the \emph{derivation} relation
$\Rightarrow$ between sequences in $V^\ast$ by $\beta
A\gamma\Rightarrow\beta\alpha\gamma$ iff $A\to\alpha$ is a production
of $G$ and $\beta,\gamma$ are arbitrary sequences in $V^\ast$.  A CFG
is \emph{acyclic}, if none of its nonterminals $A$ allows
$A\Rightarrow^+ A$.  A CFG is \emph{$\varepsilon$-free}, if none of
its productions is of form $A\to\varepsilon$ for some nonterminal $A$.

In the remainder of this section, we motivate the
problem by considering applications in computational linguistics
(\autoref{sec-mts}) and compilers construction (\autoref{sec-amb}).

\subsection{Application: Computational Linguistics}\label{sec-mts}
In contrast with many formal theories of syntax that describe natural
language sentences through `generative-enumerative means',
\citet{pullum01} champion \emph{model-theoretic syntax}, where the
syntactic structures (typically, trees) of a natural language are the
models of some logical formula.  They point out interesting
consequences on theories of syntax, but here we betray the spirit of
their work in exchange for some practicality.\footnote{In an
  ESSLLI~2013 lecture, Geoffrey Pullum famously explained that ``model
  theoretic syntax is not generative enumerative syntax with
  constraints'', the latter being exactly what we are proposing as a
  way of mitigating the complexity of model-theoretic techniques.}

Indeed, the usual approach to model-theoretic syntax would be to
describe a language through a \emph{huge} formula $\varphi$ of \PDL\
or monadic second-order logic (MSO) on trees.  Checking whether a
given sentence $w$ can be assigned a structure then reduces to a
recognition problem on a tree automaton $\?A_\varphi$ of exponential (for
\PDL) or non-elementary (for MSO) size~\citep{cornell00}.

\paragraph{A Mixed Approach}  We consider a pragmatic approach,
where
\begin{itemize}
\item a CFG describes the \emph{local} aspects of syntax, e.g.\ that a canonical
  transitive French sentence can be decomposed into a noun phrase
  acting as subject followed by a verb kernel and an object noun
  phrase corresponds to a production
  $\mathrm{S}\to\mathrm{NP}\,\mathrm{VN}\,\mathrm{NP}$, while
\item \emph{long-distance} dependencies and more complex linguistic
constraints are described through \PDL\ formul\ae.
\end{itemize}

\begin{example}[French Clitics]\label{ex-mts}
  A toy grammar for French sentences with predicative verbs like
  `dire' or `demander' could look like (in an extended syntax
  where $X?$ describes zero or one occurrences of symbol $X$):
  {\begin{align*}
    \mathrm{S}&\to\mathrm{NPsuj}?\ \mathrm{VN}\ \mathrm{VPinfobj}?\
    \mathrm{PPaobj}?\!\!\!\!\!\!
    \\\mathrm{NPsuj}&\to\mathrm{d}\ \mathrm{n} 
    \\\mathrm{VN}&\to\mathrm{clsuj}?\ \mathrm{clobj}?\ \mathrm{claobj}?\ \mathrm{v}
    \\\mathrm{VPinfobj}&\to\mathit{de}\ \mathrm{VN}
    \\\mathrm{PPaobj}&\to\text{\textit{\`a}}\ \mathrm{NP}
    \\\mathrm{v}&\to\mathit{demande}\mid\text{\textit{r\'efl\'echir}}
    &\mathrm{clsuj}&\to\mathit{elle}
    \\\mathrm{n}&\to\mathit{philosophe}
    &\mathrm{clobj}&\to\mathit{le}
    \\\mathrm{d}&\to\mathit{la}
    &\mathrm{claobj}&\to\mathit{lui}
  \end{align*}}%
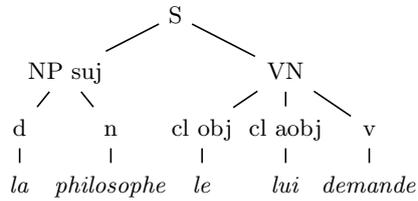
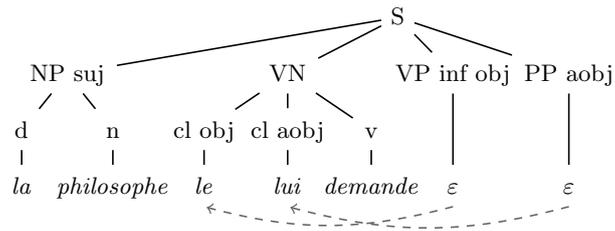
\begin{figure}[tbp]
  \centering
  \subfloat[\label{fig-a}Syntax tree according to Example~\ref{ex-mts}.]{%
    \begin{tikzpicture}[semithick,shorten >=1pt,level/.style={level
      distance=.75cm,sibling distance=1.25cm},every
      node/.style={font=\small,text height=1.5ex,text depth=.25ex}]
    \node{S}
    child[sibling distance=2.9cm]{node{NP suj}
      child[sibling distance=1.2cm]{node{d}
        child{node{\textit{la}}}}
      child[sibling distance=1.2cm]{node{n}
        child{node{\textit{philosophe}}}}}
    child[sibling distance=2.9cm]{node{VN}
      child[sibling distance=1.1cm]{node{cl obj}
        child{node(o){\textit{le}}}}
      child[sibling distance=1.1cm]{node{cl aobj}
        child{node(a){\textit{lui}}}}
      child[sibling distance=1.1cm]{node{v}
        child{node{\textit{demande}}}}};
    \node[below=-.85ex of o]{};
    \end{tikzpicture}%
  }\hspace{.3em}
  \subfloat[\label{fig-b}Analysis with moved constituents.]{%
    \begin{tikzpicture}[semithick,shorten >=1pt,level/.style={level
      distance=.75cm,sibling distance=1.25cm},every
      node/.style={font=\small,text height=1.5ex,text depth=.25ex}]
    \node{S}
    child[sibling distance=2.9cm]{node{NP suj}
      child[sibling distance=1.2cm]{node{d}
        child{node{\textit{la}}}}
      child[sibling distance=1.2cm]{node{n}
        child{node{\textit{philosophe}}}}}
    child[sibling distance=2.9cm]{node{VN}
      child[sibling distance=1.1cm]{node{cl obj}
        child{node(o){\textit{le}}}}
      child[sibling distance=1.1cm]{node{cl aobj}
        child{node(a){\textit{lui}}}}
      child[sibling distance=1.1cm]{node{v}
        child{node{\textit{demande}}}}}
    child[sibling distance=1.45cm]{node{VP inf obj}
      child[level distance=1.5cm]{node(eo){$\varepsilon$}}}
    child[sibling distance=1.5cm]{node{PP aobj}
      child[level distance=1.5cm]{node(ea){$\varepsilon$}}};
    \path[dashed,color=black!60,->]
      (eo.south) edge[bend left=15] (o.south)
      (ea.south) edge[bend left=15] (a.south);
    \end{tikzpicture}%
  }
  \caption{\label{fig-mts}Syntax trees for ``la philosophe le lui demande.''}
\end{figure}%
  Such predicative verbs have a mandatory object and subject, and an
  optional indirect object.  But all three canonical arguments can be
  replaced by \emph{clitics} in the verb matrix $\mathrm{VN}$.  This
  grammar fragment generates reasonable sentences like
  \begin{exe}
    \ex{
      \gll La philosophe demande de r\'efl\'echir.\\
      The philosopher asks to think.\\
    }
    \ex{
      \gll La philosophe le lui demande.\\
      The philosopher it.ACC her.DAT asks.\\
      \trans ``The philosopher asks it to her''.
    }
  \end{exe}\noindent
  where the `le' clitic acts as direct object and `lui' as an
  indirect one (see \autoref{fig-a} for an example syntax tree).  It
  also generates ungrammatical ones like
  \begin{exe}
    \ex[*]{
      \gll Elle le lui demande de r\'efl\'echir.\\
      She it.ACC her.DAT asks to think.\\
      \trans ``She asks it to her to think''.
    }
    \ex[*]{
      \gll demande.\\
      asks.\\
    }
  \end{exe}\noindent
  where there are duplicated or missing arguments.

  Instead of refining the grammar (which might prove impossible, for
  instance if it was automatically extracted from a treebank, i.e.\ a
  set of sentences annotated with syntactic trees), we can filter out
  the unwanted trees using a \PDL\ formula.  To improve readability,
  we take symbols like `$\mathrm{VPinfobj}$' or `$\mathrm{clsuj}$'
  to denote sets of atomic propositions, respectively
  $\{\mathrm{VP},\mathrm{inf},\mathrm{obj}\}$ and
  $\{\mathrm{cl},\mathrm{suj}\}$ in this instance, and refine our
  grammar with the following formula: 
  {\begin{align*}
  [\downarrow^\ast]\mathit{demande}&\imply\big(\tup{(\uparrow;\uparrow;\rightarrow^+)+(\uparrow;\leftarrow^+;\mathrm{cl}?)}\mathrm{obj}\tag{at
  least one object}\\
  &\qquad\wedge\tup{(\uparrow;\uparrow;\leftarrow)+(\uparrow;\leftarrow^+;\mathrm{cl}?)}\mathrm{suj}\tag{at
  least one subject}\\
  &\qquad\wedge\bigwedge_{f\in\{\mathrm{suj},\mathrm{obj},\mathrm{aobj}\}}\tup{\uparrow;\leftarrow^+;\mathrm{cl}?}f\\&\hspace{9em}\imply\neg\tup{\uparrow;\uparrow;(\leftarrow+\rightarrow^+)}f\big)\tag{clitic arguments forbid the matching canonical
  arguments} \end{align*}}

  Interestingly, such \PDL\ constraints can easily be tested against
  tree corpora to check their validity; see~\citep{lai10} on using
  \PDL-like query languages to this end.  We checked that the above
  \PDL\ formula was satisfied by the trees in the Sequoia
  treebank~\citep{sequoia} using an XPath processor: note that our
  formula is indeed in \PDLcr.
\end{example}

\paragraph{Discussion}  In this approach, the CFG can be a very
permissive, over-generating one, like the probabilistic grammars
extracted from treebanks,\footnote{\Citet{tablc} finds an
average of $7.2\times 10^{27}$ different parse trees per sentence with
a grammar extracted from the Penn treebank!} since it is later refined
by the \PDL\ constraints.  We are not aware of any linguistic
rationale for cycles in CFGs; on the other hand,
$\varepsilon$-productions are sometimes used as placeholders for \emph{moved
constituents}.  However, in such analyses, the moved constituent and
the placeholder are \emph{coindexed}, i.e.\ related through an
additional relation, which
\begin{itemize}
\item requires a richer class of models than mere
  trees over a finite alphabet if we want to make the coindexation
  explicit (see \autoref{fig-b} for an example)---one could consider
  \emph{data trees} to this end \citep{bojanczyk12,figueira12}---, and
\item can be simulated by a \PDL\ formula, as seen with the connection
  we establish between a clitic and the corresponding missing argument
  in Example~\ref{ex-mts}.
\end{itemize}
We therefore expect our grammars to be both acyclic and
$\varepsilon$-free---and we could check that this was indeed the case
on the three rather different CFGs proposed by \citet{tablc} for
natural language parsing benchmarks.

On the logical side, it seems necessary to be able to use e.g.\
depth-first traversals (recall Example~\ref{ex-gs}).  \citet{palm99}
and \citet{lai10} study the question in much greater detail and argue
that \PDLcp\ provides an appropriate expressiveness for linguistic
queries.

\subsection{Application: Ambiguity Filtering}\label{sec-amb}
Ambiguities in context-free grammars describing the syntax of
programming languages are a severe issue, as they might lead to
different semantic interpretations, and complicate the use of
deterministic parsers---they basically require manual fiddling.
They are also quite useful, as they allow for more concise and more
readable grammars, and it is actually uncommon to find a language
reference proposing an unambiguous grammar.

A nice way of dealing with ambiguities at parse time is to build a
parse forest and \emph{filter out} the unwanted trees
\citep{amb/klint}.  In contrast with tinkering with parsers, this
allows to implement the `side constraints' provided in the main text
of language references as \emph{declarative rules}, which, beyond
readability and maintainability concerns~\citep{Kats2010}, also
enables some amount of static reasoning and optimization.

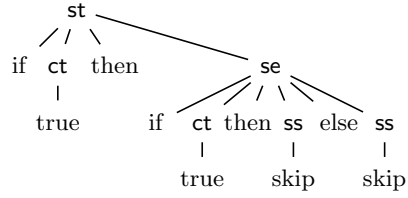
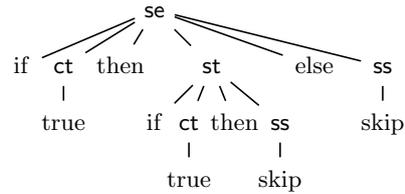
\begin{figure}[tbp]
  \centering
  \subfloat[\label{fig:shift}Parse when preferring shift over reduce.]{%
    \begin{tikzpicture}[semithick,shorten >=1pt,level/.style={level
      distance=.75cm,sibling distance=1.25cm},every
      node/.style={font=\small,text height=1.5ex,text depth=.25ex}]
    \node{$\#{st}$}
    child[sibling distance=.5cm]{node{$\mathrm{if}$}}
    child[sibling distance=.5cm]{node{$\#{ct}$}
      child{node{$\mathrm{true}$}}}
    child[sibling distance=1cm]{node{$\mathrm{then}$}}
    child[sibling distance=1.7cm]{node{$\#{se}$}
      child[sibling distance=.6cm]{node{$\mathrm{if}$}}
      child[sibling distance=.6cm]{node{$\#{ct}$}
        child{node{$\mathrm{true}$}}}
      child[sibling distance=.6cm]{node{$\mathrm{then}$}}
      child[sibling distance=.6cm]{node{$\#{ss}$}
        child{node{$\mathrm{skip}$}}}
      child[sibling distance=.6cm]{node{$\mathrm{else}$}}
      child[sibling distance=.6cm]{node{$\#{ss}$}
        child{node(a){$\mathrm{skip}$}}}};
    \node[right=.2em of a]{};
    \end{tikzpicture}%
  }\hspace{3em}
  \subfloat[\label{fig:reduce}Parse when preferring reduce over shift.]{%
    \begin{tikzpicture}[semithick,shorten >=1pt,level/.style={level
      distance=.75cm,sibling distance=1.25cm},every
      node/.style={font=\small,text height=1.5ex,text depth=.25ex}]
    \node{$\#{se}$}
    child[sibling distance=.7cm]{node{$\mathrm{if}$}}
    child[sibling distance=.8cm]{node{$\#{ct}$}
      child{node{$\mathrm{true}$}}}
    child[sibling distance=.9cm]{node{$\mathrm{then}$}}
    child[sibling distance=1.5cm]{node{$\#{st}$}
      child[sibling distance=.5cm]{node{$\mathrm{if}$}}
      child[sibling distance=.6cm]{node{$\#{ct}$}
        child{node{$\mathrm{true}$}}}
      child[sibling distance=.6cm]{node{$\mathrm{then}$}}
      child[sibling distance=.6cm]{node{$\#{ss}$}
        child{node{$\mathrm{skip}$}}}}
    child[sibling distance=1.4cm]{node{$\mathrm{else}$}}
    child[sibling distance=1.2cm]{node{$\#{ss}$}
      child{node(a){$\mathrm{skip}$}}};
    \node[right=.2em of a]{};
    \end{tikzpicture}%
  }
  \caption{\label{fig:amb}Two parses for the ambiguous input ``if true
  then if true then skip else skip'' with the grammar of Example~\ref{ex-amb}.}
\end{figure}
\begin{example}[Dangling Else]\label{ex-amb}
  We propose to use \PDL\ formul\ae\ to filter out unwanted parses.  
  Consider the following regular tree grammar for
  statements:\footnote{We use a regular tree grammar in a restricted
    way to label internal nodes differently depending on the chosen
    production; this allows for a simpler \PDL\ formula but has
    otherwise no impact as the language remains local.}
  {\begin{align*}
    S&\to \#{st}(\mathrm{if}\ C\ \mathrm{then}\ S)
      \mid\#{se}(\mathrm{if}\ C\ \mathrm{then}\ S\ \mathrm{else}\ S)
      \\&\;\;\;\,\mid\#{sw}(\mathrm{while}\ C\ S)
      \mid\#{ss}(\mathrm{skip})\\
    C&\to\#{ct}(\mathrm{true})\mid\#{cf}(\mathrm{false})
  \end{align*}}%
  Feeding this grammar to a LALR(1) parser generator like GNU/bison,
  we find a single shift/reduce conflict, where the parser has a
  choice on inputs like ``if true then if true then skip else skip'',
  upon reaching the `$\mathrm{else}$' symbol, between reading further
  (\autoref{fig:shift}), and reducing first and leaving this
  $\mathrm{else}$ for later (\autoref{fig:reduce}).  The usual
  convention in programming languages is a greedy one, where shift is
  always chosen.  However, disambiguation by choosing between shift or
  reduce parsing actions is error-prone, and there
  are cases where both alternatives are incorrect on some
  inputs (see~\citep{scico/Schmitz10} for such an example in Standard ML).
 
  A \PDL\ formula that accepts the desired tree of \autoref{fig:shift}
  but rejects the one of \autoref{fig:reduce} should check that no
  `$\mathrm{else}$' node can be next in a depth-first traversal (in
  the sense of Example~\ref{ex-gs}) from an `$\#{st}$' node:
  {\begin{align*} \neg\tup{\downarrow^\ast}(\#{st}\wedge\tup{\prec}\mathrm{else})\,.  \end{align*}}%
  Observe that a depth-first traversal $\prec$ is really needed here,
  because the `$\#{st}$' node can be at the end of an arbitrarily
  long sequence of `$\#{sw}$' nodes from nested `$\mathrm{while}$'
  statements.

  A very similar approach was proposed by \citet{amb/thorup}, who used
  simple tree patterns for similar purposes.  Both tree patterns and
  \PDL\ formul\ae\ can be compiled into the grammar, so that only the
  desired trees can be generated, allowing to use deterministic
  parsers or ambiguity checking tools~\citep{scico/Schmitz10}.  \PDL\
  formul\ae\ are strictly more expressive than patterns; the dangling
  else example required an involved extension of patterns
  in~\citep{amb/thorup96}.
\end{example}

\paragraph{Discussion}  The grammars used for programming
languages are always acyclic---tools like GNU/bison will detect and
reject cyclic grammars---but $\varepsilon$-productions
are fairly common.

On the logic side, although the formula of Example~\ref{ex-amb} is
in \PDLcp\ and not in \PDLcr, it uses depth-first traversals in a
restricted manner, and they could be expressed in XPath~1.0 as
{\small\texttt{(descendant::*|following::*)[1]}}, which selects the
first node in document order among all descendants and right siblings.
We expect \PDLcp\ to be expressive enough for most tasks,
but \emph{layout sensitive} syntax would be beyond its grasp: in
programming languages like Haskell or Python, the indentation level is
used to delimit statement blocks---differentiating between possible
parses then requires some limited counting capabilities, or infinite
label sets with order as in a recent formalization by \citet{adams13}.

Excluding a tree considered individually is one approach among others
to ambiguity filtering~\citep{amb/klint}.  A popular alternative
considers the parse forest as a whole, i.e.\ the tree automaton $\?A_{G,w}$
itself.  The ambiguity resolution of Example~\ref{ex-amb} on the input
``if true then if true then skip else skip'' can be simply stated as a
preference $\#{st}>\#{se}$ implying that the rule
{\small\begin{align*}
 (0,\!S,\!9)&\to\#{st}\big((0,\!\mathrm{if},\!1) (1,\!C,\!2) (2,\!\mathrm{then},\!3)
 (3,\!S,\!9)\big)
 \intertext{{\normalsize is preferred over the rule}}
 (0,\!S,\!9)&\to\#{se}\big((0,\!\mathrm{if},\!1) (1,\!C,\!2) (2,\!\mathrm{then},\!3)
 (3,\!S,\!7) (7,\!\mathrm{else},\!8) (8,\!S,\!9)\big)
\end{align*}}%
in the automaton $\?A_{G,w}$.  Such disambiguation rules are easy to
write, but they are also inherently \emph{dynamic}: they cannot be
compiled into the grammar, because whether the rule will be triggered
depends on whether an ambiguity appears there---an
undecidable problem.

\section{Complexity Results}\label{sec-cmpl}
We investigate in this section the complexity of the parse forest
model-checking problem.  We obtain a classification of complexities
depending on the properties of the grammar (see \autoref{fig-cmpl}).
Interestingly, our hardness results always hold for a formula
$\varphi$ in the rather restricted fragment $\PDLcr[\downarrow]$, and
generally hold already for fixed $G$ and/or $w$.  These
bounds use logarithmic space reductions.%

Turning first to the complexity in the general case, an immediate
consequence of classical results in the
field~\citep[e.g.][Theorem~7]{calvanese09} is that it lies in
\textsc{ExpTime}.
\begin{figure}[tbp]
  \centering
  \begin{tikzpicture}[semithick,shorten >=1pt,every
      node/.style={font=\small,text height=1.5ex,text depth=.25ex}]
  \node(g) at (6.3,1.6){general case};
  \node(ae) at (6.3,0){acyclic, $\varepsilon$-free};
  \node(a) at (4.3,.8) {acyclic};
  \node(e) at (8.3,.8) {$\varepsilon$-free};
  \path
    (g) edge (a)
    (a) edge (ae)
    (g) edge (e)
    (e) edge (ae);
  \path[dashed,thin,color=black!50]
    (0,.4) edge (8.7,.4)
    (0,1.2) edge (8.7,1.2);
  \node[color=black!70,text ragged,text width=3cm] at (1.5,1.6) {\textsc{ExpTime}-complete};
  \node[color=black!70,text ragged,text width=3cm] at (1.5,.8) {\textsc{PSpace}-complete};
  \node[color=black!70,text ragged,text width=3cm] at (1.5,0) {\textsc{NPTime}-complete};
  \end{tikzpicture}
  \caption{\label{fig-cmpl}The complexity of the PFMC problem, depending on the grammar characteristics.}
\end{figure}
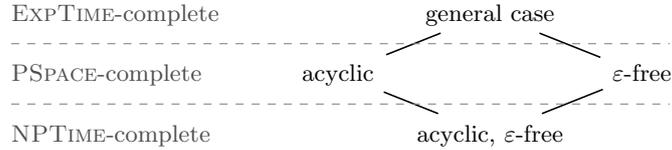
\begin{proposition}\label{prop-ete}
  PFMC is in \textsc{ExpTime}.
\end{proposition}
\begin{proof}
  One way to proceed is to \emph{localize} the automaton
  $\mathcal{A}_{G,w}$ by replacing each rule $(i,X,j)\to X(\alpha)$ by
  $(i,X,j)\to (i,X,j)(\alpha)$.  We can then apply the construction of
  Example~\ref{ex-cfg} to the resulting local automaton, thereby
  obtaining a \PDL\ formula $\varphi_{G,w}$ describing a relabeled
  parse forest of $w$ according to $G$.  It then suffices to apply the
  same relabeling to $\varphi$ by interpreting each atomic proposition
  $X$ as $\bigvee_{1\leq i\leq j\leq n}(i,X,j)$, yielding $\varphi'$,
  and to use the \textsc{ExpTime} upper bound on \PDL\ satisfiability
  of \citet{pdltree} to the conjunction $\varphi'\wedge\varphi_{G,w}$
  to conclude.

  An issue with this proof is that it yields an exponential complexity
  even if $\varphi$ is fixed.  We can improve on this using
  automata-based techniques: assume $G$ to be in quadratic form and
  $\varphi$ to work on binary trees that encode unranked trees with
  the $\swarrow$ and $\rightarrow$ relations from
  Example~\ref{ex-nav}, as these transformations only incur a linear
  cost.  Then, construct the tree automaton $\mathcal{A}_{G,w}$ of
  size $O(|G|\cdot|w|^3)$ that recognizes the set of parse trees of
  $w$ in $G$ and the tree automaton $\mathcal{A}_\varphi$ of size
  $2^{p(|\varphi|)}$ for a polynomial $p$ that recognizes the models
  of $\varphi$: it suffices to test the emptiness of their product
  automaton, which can be performed in time linear in
  $|G|\cdot|w|^3\cdot 2^{p(|\varphi|)}$ for a polynomial $p$.
\end{proof}
An interesting consequence of the proof of Prop.~\ref{prop-ete}
is that the PFMC problem is \textsc{PTime}-complete when the \PDL\
formula is fixed, pleading for using small formul\ae\ in practice.

Our proof for Prop.~\ref{prop-ete} does not benefit from the
specificities of the PFMC problem: any satisfiability problem in
presence of a tree language would use the same algorithm.  Therefore,
we might still hope for the existence of a more efficient solution,
but adapting the proof of \textsc{ExpTime}-hardness for PDL
satisfiability from \citep{blackburn01}, we obtain:
\begin{proposition}
  PFMC is \textsc{ExpTime}-hard, even for fixed $G$ and $w$ and
  for $\varphi$ in $\PDLcr[\downarrow]$.
\end{proposition}
\begin{proof}[Proof Idea]
  We go back to low-level arguments\footnote{ One could attempt to
    reduce from the satisfiability problem for
    ${\PDL[{\downarrow}]}$---which is
    \textsc{ExpTime}-complete~\citep{pdltree}---, but it seems to us
    that such a proof would require changing the structure of the
    satisfaction witnesses by adding $\varepsilon$-leaves, and we do
    not see any straightforward way of handling this modification in
    the formula.}  and reduce from the \emph{two-players corridor
    tiling game} of \citet{chlebus86}.  We fix $w=\varepsilon$ and
  also fix $G$ to generate a parse forest encoding game trees; we use
  a $\PDLcr[\downarrow]$ formula $\varphi$ to check that there exists
  a winning strategy.  See \appref{ax-gal} for details.
\end{proof}
As can be seen from this proof idea, the fact that $w=\varepsilon$ and
$G$ is cyclic plays an important role, because the parse forest is
essentially unconstrained.  This is a good incentive to examine what
happens when $G$ is acyclic and/or $\varepsilon$-free, especially
since those cases are most relevant for the applications we described
in \autoref{sec-mcpf}.

\subsection{The Acyclic {\boldmath $\varepsilon$}-free Case: Mixed
 Model-Theoretic Syntax}
Let us therefore consider the other end of our spectrum, which we
claimed was of particular relevance for the mixed approach to
model-theoretic syntax we presented in \autoref{sec-mts}: if $G$ is
acyclic and $\varepsilon$-free, then $\?A_{G,w}$ is a non-recursive
tree automaton generating a \emph{finite} parse forest, albeit it might
contain exponentially many trees.  This yields an \textsc{ExpTime}
algorithm that performs \PDL\ model-checking (in
\textsc{Ptime}~\citep{lange06}) on each tree individually.  We can try
to refine this first approach and resort
to \citep[Lemma~7.5]{benedikt08}, which entails that the problem for
the \PDLcr\ fragment is in \textsc{PSpace}, but we can do better:
\begin{proposition}\label{prop-aef}
  PFMC with acyclic and $\varepsilon$-free grammars is
  \textsc{NPTime}-complete; hardness holds even
  for fixed $G$ and for $\varphi$ in $\PDLcr[\downarrow]$.
\end{proposition}
\begin{proof}[Proof of the Upper Bound]
  We show that the parse trees in $L_{G,w}=L(\?A_{G,w})$ are of polynomial
  size in $|G|$ and $|w|$.  The \textsc{NPTime} algorithm then guesses a tree in
  $L_{G,w}$  and checks that it is a model in polynomial
  time~\citep{lange06}.

  \begin{claim}
    Let $G =\tup{N,\Sigma,P,S}$ be an acyclic and $\varepsilon$-free
    CFG.  Let $w\in\Sigma^\ast$.  Any parse tree $t$ in $L_{G,w}$ has
    at most $|N|(|w|-1) + |w|$ nodes.
  \end{claim}
  Consider the run of $A_{G,w}$ on $t$: each node of $t$ is labeled by
  a state $(i,A,j)$ describing two positions $0\leq i\leq j\leq n$ in
  $w$ and a nonterminal $A$ in $N$.  Because $G$ is
  $\varepsilon$-free, we know that
  $i<j$.  %
  We claim that the set of nodes labelled with positions $(i,j)$ forms
  a connected chain.

  To see this, suppose two nodes $a$ and $b$ are both labelled with
  positions $(i,j)$.  Suppose first that neither $a$ nor $b$ is an
  ancestor of the other.  Let then $c\not\in\{a,b\}$ be their least
  common ancestor (lca): $c$ must have at least two children, and its
  children will be labelled with non-overlapping positions---recall
  that $i<j$.  Only one of these non-overlapping intervals can contain
  the interval $(i,j)$.  The child corresponding to that interval
  would then be the lca of $a$ and $b$, in contradiction with $c$
  being their lca: hence one of $a$ or $b$ is the lca of $a$ and $b$.

  Suppose now without loss of generality that $a$ is an ancestor of
  $b$.  Observe that a descendant of $a$ would be labelled with a
  sub-interval of $(i,j)$, and an ancestor of $b$ would be labelled
  with a super-interval of $(i,j)$.  This forces every node in the
  path from $a$ to $b$ to be labelled with $(i,j)$.  Hence, the nodes
  labelled with $(i,j)$ form a connected chain.

  Since $G$ is acyclic, each chain of nodes
  $(i,A_1,j),(i,A_2,j)\cdots(i,A_p,j)$ having the same positions
  $(i,j)$ cannot have a non-terminal $A_k$ occuring twice, or the
  grammar would allow a cycle.  Therefore, each such chain will have
  at most $|N|$ nodes.  We can `collapse' these chains to form a
  tree where each $(i,j)$ pair appears at most once, and every node
  (except the leaves) has at least two children.  Since there are
  exactly $|w|$ leaves ($G$ is $\varepsilon$-free), there can be at
  most $|w|-1$ internal nodes in such a tree.  We obtain that there
  were at most $|N|(|w|-1)$ internal nodes in the original parse tree,
  i.e.\ at most $|N|(|w|-1) + |w|$ nodes in the full parse tree.
\end{proof}
\begin{proof}[Proof Idea for the Lower Bound]
  We reduce from 3SAT with a fixed grammar $G$ and a
  \mbox{$\PDLcr[\downarrow]$} formula $\varphi$; see
  \appref{ax-nphard} for details.
\end{proof}

\subsection{Non-Recursive DTDs}

Let us turn now to the more involved cases where $G$ is either acyclic
or $\varepsilon$-free: we rely in both cases for the upper bounds on
the same result that extends Lemma~7.5 of \citet{benedikt08} to handle
\PDL\ instead of \PDLcr:
\begin{proposition}\label{prop-dtd}
  Satisfiability of \PDL\ in presence of a non-recursive DTD is
  \textsc{PSpace}-complete.
\end{proposition}
In this proposition, a \emph{document type definition} (DTD) is a
generalized CFG $D=\tup{N,P,S}$ where $P$ is a mapping from $N$
to \emph{content models} in $\mathrm{Reg}(N^\ast)$ the set of regular
languages over $N$---we will assume these content models
to be described by finite automata (NFA).  Given $D$, the derivation
relation $\Rightarrow$ relates $\beta A\gamma$ to $\beta\alpha\gamma$
iff $\alpha$ is in $P(A)$; a DTD is \emph{non-recursive} if no
nonterminal has a derivation $A\Rightarrow^+\beta A\gamma$ for some
$\beta,\gamma$ in $N^\ast$.  Note that a non-recursive DTD might still
generate an infinite tree language, but that all its trees will have a
depth bounded by $|N|$.
\begin{proof}[Proof Idea for Prop.~\ref{prop-dtd}]
  The hardness part is proven by \citet{benedikt08} in their
  Prop.~5.1.

  For the upper bound, we reduce to the emptiness problem of a 2-way
  alternating parity \emph{word} automaton, which is
  in \textsc{PSpace}~\citep{serre06}.  The key idea, found
  in \citeauthor{benedikt08}'s work, is to encode trees of bounded
  depth as XML strings (i.e.\ with opening and closing tags): both the
  DTD $D$ and the formula $\varphi$ can then be encoded as alternating
  parity word automata $\?A_D$ and $\?A_\varphi$ of polynomial size.
  Because we handle the full \PDL\ instead of only \PDLcr, our
  construction for $\?A_\varphi$ has to extend that
  of \citeauthor{benedikt08}---for instance, we cannot assume
  $\?A_\varphi$ to be loop-free.  See \appref{ax-dtd} for
  details.
\end{proof}

\subsection{Acyclic Case: Ambiguity Filtering}
We are now ready to attack the case of acyclic grammars.  This
restriction is enough to ensure that the parse forest is finite, and,
more importantly, $\?A_{G,w}$ is trivially non-recursive, thus
Prop.~\ref{prop-dtd} immediately yields a \textsc{PSpace} upper
bound.  In fact, this is optimal:
\begin{proposition}\label{prop-a}
  PFMC with acyclic grammars is \textsc{PSpace}-complete; hardness
  holds even for fixed $w$ and for $\varphi$ in $\PDLcr[\downarrow]$.
\end{proposition}
\begin{proof}%
  Because $G$ is acyclic, for any $w$, the trimmed version of
  $\?A_{G,w}$ is a non-recursive tree automaton.  Indeed, in this
  automaton, if a state $(i_0,A,i_k)$ of $\?A_{G,w}$ rewrites in $n$
  steps into a tree $t$ with leaves labeled by
  $(i_0,X_1,i_1)\cdots(i_{k-1},X_k,i_k)$, then by induction on $n$,
  $A\Rightarrow^n X_1\cdots X_k$ in $G$.  If the automaton is trim,
  then the existence of a state $(i,B,j)$ implies that $B$ derives the
  factor $a_{i+1}\cdots a_{j}$ of $w$.  Thus, if $(i,A,j)$ were to
  rewrite in at least one step into a tree $C[(i,A,j)]$ in the trimmed
  $\?A_{G,w}$, then the tree $C[\varepsilon]$ has $\varepsilon$ as
  yield, and there would be a cycle $A\Rightarrow^+ A$ in $G$, a
  contradiction.  It remains to localize $\?A_{G,w}$ by relabelling
  the rules $(i,A,j)\to A(q_1\cdots q_m)$ as $(i,A,j)\to
  (i,A,j)(q_1\cdots q_m)$ to obtain a \emph{local} non-recursive tree
  automaton, which is just a particular case of a non-recursive DTD,
  and interpret the propositions $p$ in $V'=\#{AP}$ as $\bigvee_{0\leq
  i\leq j\leq n}(i,p,j)$ over $Q_{G,w}$ in $\varphi$ to apply
  Prop.~\ref{prop-dtd} and obtain the upper bound.

  The lower bound holds for the $\PDLcr[\downarrow]$ satisfiability
  problem in presence of non-recursive and no-star
  DTDs~\citep[Prop.~5.1]{benedikt08}, which is easy to reduce to our
  problem by simply adding $\varepsilon$-leaves in the DTD; this
  lower bound thus already holds for a fixed $w=\varepsilon$.
\end{proof}

\subsection{{$\varepsilon$}-Free Case: {\PDL\ }Recognition}
A key question if
model-theoretic syntax is to be used in practice for natural language
processing is the following \emph{recognition problem}:
\begin{pp}\hfill
\begin{description}
\item[input] a \PDL\ formula $\varphi$, a word $w$ in $\#{AP}^\ast$,
and a distinguished proposition $s$ in $\#{AP}$,
\item[question]\IEEEhspace{1em} does there exist a tree $t$ with yield $w$ and root
label $s$ s.t.\ $t\models\varphi$?\qed
\end{description}
\end{pp}
Note in particular that the statement of the problem excludes
$\varepsilon$-labeled leaves, which would require a different formulation and
would yield an \textsc{ExpTime}-complete problem.

The recognition problem motivates the last case of our study: Due to
cycles, an $\varepsilon$-free grammar $G$ can have infinitely many
parses for a given input string $w$, and its parse trees unbounded
depth.  Nevertheless, recursions in a parse forest of an
$\varepsilon$-free grammar display a particular shape: they
are \emph{chains} of unit rules $q_i\to A_i(q_{i+1})$.  The key idea
here is that such chains define regular languages of single-strand
branches, which can be encoded in a non-recursive DTD by `rotating'
them, i.e.\ seeing the chain as a siblings sequence instead of a
parents sequence, taking advantage of the DTD's ability to describe
trees of unbounded rank.  This hints at a reduction to
the \textsc{PSpace} algorithm of Prop.~\ref{prop-dtd}.

\subsubsection{PFMC in the $\varepsilon$-free Case}
We want to reduce the problem to the satisfiability
problem for $\PDL$ in presence of a non-recursive DTD and use
Prop.~\ref{prop-dtd}.  Our algorithm starts by constructing
$\?A_{G,w}$ in polynomial time on binarized trees.  As in the proof of
Prop.~\ref{prop-a}, we consider `localized' rules $q\to
q(q_1\,q_2)$ of $\?A_{G,w}$, and replace them by productions of the
form $q\to\#{chains}(q_1)\#{chains}(q_2)$ where the
$\#{chains}(q_i)$ are the languages of single chains out of $q_i$.
By suitably labeling our trees, we can interpret $\varphi$ over
those transformed trees.

\begin{proposition}\label{prop-efpe}
  PFMC with $\varepsilon$-free grammars is in \textsc{PSpace}.
\end{proposition}
\begin{proof}
Let $G=\tup{N,\Sigma,P,S}$, $w$ be a string in $\Sigma^\ast$, and
$\varphi$ be a \PDL\ formula.  Without loss of generality, we assume
$G$ to have productions with right-parts of length at most 2; since
$G$ is $\varepsilon$-free, these right-parts have length at least
one.  We want to construct a non-recursive DTD $D=\tup{N',P',S'}$ and
a \PDL\ formula $\varphi'$ s.t.\ the parse forest model checking
problem on $G$, $w$, and $\varphi$ has a solution iff $\varphi$ is
satisfiable in presence of $D$, thereby reducing our instance to an
instance of a problem in \textsc{PSpace} by
Prop.~\ref{prop-dtd}.

We build $D$ from the polynomial-sized automaton
$\mathcal{A}_{G,w}$ by removing chains of \emph{unit} rules $q\to
A(q')$ of $\mathcal{A}_{G,w}$; recall that this automaton uses states
of form $q=(i,X,j)$ where $0\leq i<j\leq |w|$ and $X$ is in $V$.  Let
for this $\bar{Q}_{G,w}$ be a disjoint copy of $Q_{G,w}$ and define
$N'\eqdef\bar{Q}_{G,w}\uplus Q_{G,w}$.

\paragraph{Chain Sequences}
For each $q$ in $Q_{G,w}$, we consider the set of sequences
of successive states $q=q_0,q_1,\dots,q_n$ we can visit using only
unit rules $q_i\to A_i(q_{i+1})$ of $\delta_{G,w}$ and such that $q_n$
has a binary rule $q_n\to A_n(q'\,q'')$ or a nullary
rule $q_n\to a()$ in $\delta_{G,w}$.  More
precisely, we are interested in the relabeled sequence
$\bar{q}=\bar{q}_0,\bar{q}_1,\dots,\bar{q}_{n-1},q_n$ of copies
$\bar{q}_i$ of $q_i$, except on the very last position.  We call $\#{chains}(q)$
the language of such sequences.  Formally, $\#{chains}(q)$ is a regular language
over $N'$ that we can define thanks to a
NFA $\?A_q\eqdef\tup{N',N',\delta_q,\{\bar{q}\},Q_{G,w}}$ with state
space $N'$ where
\begin{align*}
  \delta_q&\eqdef\!\{(\bar{p},\bar{p},\bar{p}')\mid\exists A\in N,p\to
  A(p')\in\delta_{G,w}\}\\
  &\cup\{(\bar{p},p,p)\mid\exists A\in\! N,\exists p_1,p_2\!\in\!
  Q_{G,w},p\to A(p_1p_2)\!\in\!\delta_{G,w}\}\\
  &\cup\{(\bar{p},p,p)\mid\exists a\in\Sigma,p\to a()\in\delta_{G,w}\}\;.
\end{align*}
Note that $\?A_q$ has a size linear in that of $\?A_{G,w}$.  We can
see $\#{chains}$ as a regular substitution from $Q_{G,w}^\ast$ to ${N'}^\ast$
by setting $\#{chains}(\varepsilon)\eqdef\varepsilon$ and
$\#{chains}(uv)\eqdef\#{chains}(u)\#{chains}(v)$ for all $u,v$ in
$Q_{G,w}^\ast$.

\begin{figure}[tbp]
  \centering
  \subfloat[\label{fig-arun}A run of $\?A_{Q,w}$ on a tree in a
    parse forest of $w$ according to $G$.]{%
    \begin{tikzpicture}[semithick,shorten >=1pt,level/.style={level
      distance=.75cm,sibling distance=3cm},every
      node/.style={font=\small,text height=1.5ex,text depth=.25ex}]
    \node{$q_0,A_0$}
    child{node{$q_1,A_1$}
      child{node{$q_2,A_2$}
        child{node{$q_3,A_3$}
          child[sibling distance=1cm]{node{$q_4,a_1$}}
          child[sibling distance=1cm]{node{$q_5,a_2$}}}}}
    child{node{$q_6,A_6$}
      child{node{$q_7,A_7$}
        child[sibling distance=2cm]{node{$q_8,A_8$}
          child[sibling distance=1cm]{node{$q_9,a_3$}}
          child[sibling distance=1cm]{node{$q_{10},a_4$}}}
        child[sibling distance=2cm]{node{$q_{11},A_9$}
          child{node{$q_{12},a_5$}}}}};
    \end{tikzpicture}}\linebreak
  \subfloat[\label{fig-trrun}The corresponding tree in $L(D)$.]{%
    \begin{tikzpicture}[semithick,shorten >=1pt,level/.style={level
          distance=.75cm,sibling distance=1cm},every
        node/.style={font=\small,text height=1.5ex,text depth=.25ex}]
      \node{$q_0$}
      child{node{$\bar{q}_1$}}
      child{node{$\bar{q}_2$}}
      child{node{$q_3$}
        child[sibling distance=.5cm]{node{$q_4$}}
        child[sibling distance=.5cm]{node{$q_5$}}}
      child{node{$\bar{q}_6$}}
      child{node{$q_7$}
        child[sibling distance=.5cm]{node{$q_8$}
          child{node{$q_9$}}
          child{node{$q_{10}$}}}
        child[sibling distance=.5cm]{node{$\bar{q}_{11}$}}
        child[sibling distance=.5cm]{node{$q_{12}$}}};
  \end{tikzpicture}}
  \caption{\label{fig-todtd}The tree transformation for the proof of Prop.~\ref{prop-efpe}.}
\end{figure}
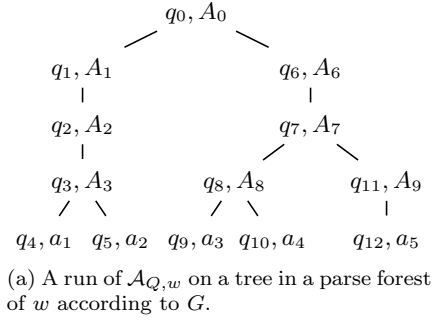
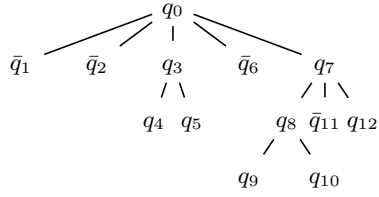
\paragraph{The DTD}  We can now express the productions $P'$ of $D$:
\begin{align*}
  P(\bar{q})&\eqdef\varepsilon&
  P(q)&\eqdef\bigcup_{q\to A(q_1\,
    q_2)\in\delta_{G,w}}\hspace{-1em}\#{chains}(q_1\,q_2)
\cup\bigcup_{q\to a()\in\delta_{G,w}}\hspace{-1em}\varepsilon\;.
\end{align*}
Thus, the symbols in $\bar{Q}_{G,w}$ are leaves and only
employed to represent a chain sequence that has been transformed into
a sequence of siblings in the DTD.  Also, because any word in
$\#{chains}(q)$ for some $q$ is of form $up$ with $u$ in
$\bar{Q}_{G,w}^\ast$ and $p$ in $Q_{G,w}$, any internal node in a tree
of $D$ has exactly two children labeled by states in $Q_{G,w}$.
Therefore, and because $G$ is $\varepsilon$-free, we get that $D$ is
non-recursive.  See \autoref{fig-todtd} for an illustration of the
tree transformation we operated.

\paragraph{The Formula}  It remains to define a formula
$\varphi'$ that will be interpreted on the transformed trees of $D$.
For this, we need to interpret the atomic propositions in $\#{AP}=V$
over the new set of labels $N'$, and to interpret the child
$\downarrow$ and sibling $\rightarrow$ relations.

Regarding the atomic propositions, we can interpret a label $X$ in
$V$ as
\begin{equation*}
  \bigvee_{0\leq i<j\leq |w|}(i,X,j)\vee\overline{(i,X,j)}
  \tag{interpretation of $X$}
\end{equation*}
over $N'$.

Regarding the relations, we first define
$\#{bar}\eqdef\bigvee_{\bar{q}\in\bar{Q}_{G,w}}\bar{q}$ to help us
differentiate between `rotated' nodes and preserved ones.  We then
interpret $\downarrow$ as a disjunction of paths depending on whether
we are on a rotated node, where the test $[\leftarrow]\neg\#{bar}$
allows to check that the current node is an `original' child:
\begin{gather*}
  (\neg\#{bar}?;\downarrow;[\leftarrow]\neg\#{bar}?) +
  (\#{bar}?;\rightarrow)\tag{interpretation of $\downarrow$}
\intertext{For $\rightarrow$, we set:}
  ([\leftarrow]\neg\#{bar}?);(\#{bar}?;\rightarrow)^\ast;\neg\#{bar}?;\rightarrow
  \;.\tag{interpretation of $\rightarrow$}
\end{gather*}
The initial test prevents the nodes taken from a chain from having a
right sibling; then the test sequence advances to the end of the chain
before we make the actual move to the original right sibling.

We can conclude by first noting that both $D$ and $\varphi'$ can be computed
in logarithmic space in the size of the input, and then by invoking
Prop.~\ref{prop-dtd}.
\end{proof}

Prop.~\ref{prop-efpe} is optimal:
\begin{proposition}\label{prop-efph}
  PFMC with $\varepsilon$-free grammars is \textsc{PSpace}-hard, even
  for fixed $G$ and $w$ and for $\varphi$ in $\PDLcr[\downarrow]$.
\end{proposition}
\begin{proof}[Proof Idea]
  The proof is by reduction from membership in a linear bounded
  automaton.  We fix $w=a$ for some symbol $a$ of $\Sigma$, and also
  fix the CFG $G$ to basically generate any single-strand tree with a
  root $S$ and a leaf $a$ over a fixed alphabet.  A
  $\PDLcr[\downarrow]$ formula of polynomial size then checks that
  this tree encodes an accepting run of the LBA.  See
  \appref{ax-efhard} for details.
\end{proof}

\subsubsection{{\PDL} Recognition}
The previous approaches to the recognition problem have used tree automata
techniques~\citep[e.g.][]{cornell00} or tableau-like
techniques~\citep{palm01}.  In both cases, exponential time upper
bounds were reported by the authors---to be fair, these algorithms
solve the \emph{parsing} problem and 
find a representation of \emph{all} the parses for $w$ compatible with
$\varphi$---, but we can improve on this thanks to Prop.~\ref{prop-efpe}:
\begin{corollary}\label{coro-mts}
  \PDL\ recognition is \textsc{PSpace}-complete; hardness holds even for
  fixed $w$ and for $\varphi$ in $\PDLcr[\downarrow]$.
\end{corollary}
\begin{proof}[Proof Sketch]
  The lower bound stems from an easy reduction from
  Prop.~\ref{prop-efph}: we can encode the
  grammar $G$ into a $\PDLcr[\downarrow]$ formula $\varphi_G$ as in
  Example~\ref{ex-cfg} and reduce to the recognition problem for $w$ and
  $\varphi\wedge\varphi_G$.

  For the upper bound, we can assume as usual $\varphi$ to work on a
  binary encoding of trees.  The idea is to reduce to the PFMC problem
  with a `universal' CFG that accepts all the trees of rank at most
  2 over $\#{AP}$.  A smallish issue is that we need to separate between
  nonterminal and terminal labels, but we can create a disjoint copy
  $N\eqdef\{P\mid p\in\#{AP}\}$ of $\#{AP}$ and interpret
  $\varphi$ as a formula over $N\uplus\#{AP}$ with $P\vee p$ as
  the interpretation of $p$.  This grammar has then $S$ as axiom and
  productions $A\to X\ Y$ and $A\to X$ for all $A$ in $N$ and $X,Y$ in
  $V$, and we can resort to Prop.~\ref{prop-efpe} to conclude.
\end{proof}


\section{Conclusion}\label{sec-concl}
Because \PDL\ formul\ae\ can freely navigate in trees, properties that
rely on long-distance relations are convenient to express, in contrast
with the higly local view provided by a grammar production.  However,
this expressiveness comes at a steep price, as recognition problems
using \PDL\ are \textsc{ExpTime}-complete instead
of \textsc{PTime}-complete on CFGs.

The \PDL\ model-checking of the parse trees of a CFG allows to mix the
two approaches, using a grammar for the bulk work of describing trees
and using more sparingly a \PDL\ formula for the fine work.  We argue
that this trade-off finds natural applications in computational
linguistics and compilers construction, where sensible restrictions on
the grammar lower the complexity to \textsc{NPTime}
or \textsc{PSpace}.

An additional consequence is that the recognition problem for \PDL\ is
in \textsc{PSpace}.  This is a central problem in model-theoretic
syntax, and this lower complexity suggests that `lazy' approaches,
in the spirit of the tableau construction of
\citet{palm01}, should perform significantly better than the automata
constructions of \citet{cornell00}.

More broadly, we think that our initial investigations of
model-checking parse trees open the way to a new range of applications
of model-checking techniques on parse structures and grammars.  In
particular, as pointed out in our examples, both in computational
linguistics and in ambiguity filtering for programming languages,
there are incentives to look at classes of models that generalize
finite trees over a finite label set---we expect furthermore these
generalizations to differ from the data tree model found in the XML
literature.

\section*{Acknowledgments}

The authors thank the anonymous reviewers for their helpful comments;
in particular for pointing out that the $\PDL$ formula in
Example~\ref{ex-amb} could be expressed in XPath 1.0.

\bibliography{journals,conferences,pdlcfl}
\clearpage
\appendix\pagenumbering{roman}\renewcommand{\leftmark}{Appendices}
\section{General Case}\label{ax-gal}
PDL satisfiability is known to be \textsc{ExpTime}-complete in general
\citep{fischer79}.  The general case of the parse forest
model-checking problem, i.e.\ when $G$ is an arbitrary grammar, is
also \textsc{ExpTime}-complete.  The upper bound follows from
classical techniques~\citep{calvanese09}---see the
proof sketch of Proposition~\ref{prop-ete}.

The lower bound could be proven by a reduction from PDL
satisfiability
using a `universal' CFG as in the proof of Corollary~\ref{coro-mts}.
However, this proof does not lend itself very easily to the restricted
case we want to consider, where $w$ and $G$ are fixed and $\varphi$ is
a downward $\PDLcr[\downarrow]$ formula.  We present in this section a
reduction from the \emph{two-player corridor game}, which is known to
be \textsc{ExpTime}-hard~\citep{chlebus86}, adapted from a similar
proof for the hardness of PDL satisfiability by
\citet[Theorem~6.52]{blackburn01}.

\paragraph{Two Player Corridor Game} sees two players, Eloise and
Abelard, compete by tiling a corridor.  The tiles are squares
decorated by $s+2$ different patterns $T=\{t_0,\dots,t_{s+1}\}$; two
binary relations $U$ and $R$ over $T$ tell if a tile can be placed on
top of the other and to the right of the other.  Two tiles are
distinguished: $t_0$ is called the \emph{white} tile and $t_{s+1}$ the
\emph{winning} tile.  The corridor is made of $n+2$ columns of
infinite height, with the first and last columns filled with white
tiles $t_0$ and delimiting $n$ columns for the play.  The initial
bottom row is tiled by a sequence $I_1\cdots I_n$ of tiles, which is
assumed to be correct, i.e.\ to respect the $R$ relation.

The players alternate and choose a next tile in $T$ and place it in
the next position, which is the lowest leftmost free one---thus the
chosen tile should match the tile to its left (using $R$) and the tile
below (using $U$)---; see \autoref{fig-game-pos}.  Eloise starts the
game and wins if after a finite number of rounds, the winning tile
$t_{s+1}$ is put in column $1$.  Given an instance of the 2-players
corridor tiling game, i.e.\ $\tup{s+2,I_1\cdots I_n,R,U}$, deciding
whether Eloise has a \emph{winning strategy}, i.e.\ a way of winning
no matter what Abelard plays, is \textsc{ExpTime}-complete.

\paragraph{Notation}  We represent strategy trees as parse trees.
Our $\PDLcr[\downarrow]$ formula $\varphi$ will ensure that the parse
tree is indeed a valid game tree, and that it encodes a winning
strategy for Eloise.

\renewcommand{\floatpagefraction}{.7}
\begin{figure}[tbp]
  \centering
  \subfloat[\label{fig-game-pos}One of Eloise's turns in the game.]{%
  \begin{tikzpicture}[semithick,shorten >=1pt,node distance=.1cm,every
      node/.style={font=\footnotesize,text height=2ex,text
        depth=.25ex,text centered}]
      \node[draw,text width=3.4ex](c00){$t_0$};
      \node[draw,above=of c00,text width=3.4ex](c01){$t_0$};
      \node[draw,above=of c01,text width=3.4ex](c02){$t_0$};
      \node[draw,above=of c02,text width=3.4ex](c03){$t_0$};
      \node[above=of c03,text width=3.4ex]{$\vdots$};
      \node[below=of c00]{$C_0$};
      \node[draw,right=of c00,text width=3.4ex](c10){$I_1$};
      \node[draw,above=of c10,text width=3.4ex](c11){};
      \node[draw,above=of c11,text width=3.4ex](c12){$t_{k_1}$};
      \node[below=of c10]{$C_1$};      
      \node[draw,right=of c10,text width=3.4ex](c20){$I_2$};
      \node[draw,above=of c20,text width=3.4ex](c21){};
      \node[draw,above=of c21,text width=3.4ex](c22){$t_{k_2}$};
      \node[below=of c20]{$C_2$};
      \node[right=of c20](d0){$\cdots$};
      \node[right=of c21](d1){$\cdots$};
      \node[right=of c22](d2){$\cdots$};   
      \node[draw,right=of d0,text width=3.4ex](cim0){$I_{i-1}$};
      \node[draw,above=of cim0,text width=3.4ex](cim1){};
      \node[draw,above=of cim1,text width=3.4ex](cim2){$t_{k_{i-1}}$};
      \node[below=of cim0]{$C_{i-1}$};
      \node[draw,right=of cim0,text width=3.4ex](ci0){$I_{i}$};
      \node[draw,above=of ci0,text width=3.4ex](ci1){$t_{k_{i}}$};
      \node[below=of ci0]{$C_{i}$};
      \node[draw,right=of ci0,text width=3.4ex](cip0){$I_{i+1}$};
      \node[draw,above=of cip0,text width=3.4ex](cip1){$t_{k_{i+1}}$};
      \node[below=of cip0]{$C_{i+1}$};
      \node[right=of cip0](e0){$\cdots$};
      \node[right=of cip1](e1){$\cdots$};
      \node[draw,right=of e0,text width=3.4ex](cn0){$I_{n}$};
      \node[draw,above=of cn0,text width=3.4ex](cn1){$t_{k_{n}}$};
      \node[below=of cn0]{$C_{n}$};
      \node[draw,right=of cn0,text width=3.4ex](cnp0){$t_0$};
      \node[draw,above=of cnp0,text width=3.4ex](cnp1){$t_0$};
      \node[draw,above=of cnp1,text width=3.4ex](cnp2){$t_0$};
      \node[draw,above=of cnp2,text width=3.4ex](cnp3){$t_0$};
      \node[above=of cnp3,text width=3.4ex]{$\vdots$};
      \node[below=of cnp0]{$C_{n+1}$};   
      \node[draw,above right=1cm and 1cm of cim2,text width=3.4ex](t){$t_j$};
      \path[->,auto,color=black!50,dashed] (cim2.east) edge node{$R?$}
      (t.west) (ci1.north) edge[swap] node{$U?$} (t.south);
      \node[above=1cm of d2,text ragged]{Eloise's turn, move $a$:};
  \end{tikzpicture}}\linebreak
  \subfloat[\label{fig-game-tree}The tree encoding of the turn.]{%
    \begin{tikzpicture}[node distance=.1cm,semithick,shorten
      >=1pt,level/.style={level distance=.75cm,sibling
        distance=1.25cm},every node/.style={font=\footnotesize,text
        height=1.5ex,text depth=.25ex}]
      \node(X){$X$};
      \node[below left=.8cm and 4.3cm of X](M){$M$};
      \node[below left=.15cm and .15cm of M](X1){$X$};
      \node[below=.15cm of M](M1){$M$};
      \node[below left=.15cm and .15cm of M1](X2){$X$};
      \node[below=.6cm of M1](M2){$M$};
      \node[below left=.15cm and .15cm of M2](X3){$X$};
      \node[below=.15cm of M2](M3){$M$};
      \node[below=.15cm of M3](Me){$\varepsilon$};
      \node[left=-.8cm of X2]{\begin{minipage}{1cm}%
          \[\left\{\hspace*{-1.5ex}\begin{array}{c}\\[2cm]\end{array}\right.\]%
          \end{minipage}};
      \node[below left=-.8cm and .5cm of X2,rotate=90]{next moves};
      \path (X) edge (M.north) 
      (M) edge (X1.north) (M) edge (M1.north) 
      (M1) edge (X2.north) (M1) edge[dotted] (M2.north)
      (M2) edge (X3.north) (M2) edge (M3.north)
      (M3) edge (Me);
      \node[right=of M](L){$L$};\node[below=.15cm of L](Le) {$\varepsilon$};\path (L) edge (Le);
      \node[right=of L](E){$E$};\node[below=.15cm of E](Ee) {$\varepsilon$};\path (E) edge (Ee);
      \node[below right=-.15cm and -.1cm of E](P){\begin{minipage}{1cm}%
          \[\left\{\hspace*{-1.5ex}\begin{array}{c}P\\\vdots\\P \end{array}\right.\]%
          \end{minipage}};\node[below=.1cm of P](pe){};\node[below=.8cm of P](Pe) {$\varepsilon$};\path (pe) edge (Pe);
      \node[below left=-.5cm and .1cm of P,rotate=90]{$i$ times};
      \node[below right=0cm and .9cm of E](T){\begin{minipage}{1cm}%
          \[\left\{\hspace*{-1.5ex}\begin{array}{c}T\\T\\\vdots\\T \end{array}\right.\]%
          \end{minipage}};\node[below=.1cm of T](te){};\node[below=.8cm of T](Te) {$\varepsilon$};\path (te) edge (Te);
      \node[below left=-.9cm and 0cm of T,rotate=90]{$j+1$ times};
      \node[right=4.1cm of E](C0){$C$};
      \node[below right=0cm and -.1cm of C0](C){$C$};
      \node[below right=0cm and 1cm of C0](Q){\begin{minipage}{1cm}%
          \[\left.\begin{array}{c}1\\0\\\vdots\\1 \end{array}\hspace*{-1.5ex}\right\}\]%
          \end{minipage}};
      \node[below right=-1cm and .1cm of Q,rotate=-90]{$a$ in
        binary};\node[below=.1cm of Q](qe){};\node[below=.8cm of Q](Qe) {$\varepsilon$};\path (qe) edge (Qe);
      \path (X) edge (3,-1);
      \node[below left=.5cm and .8cm of C](T1){\begin{minipage}{1cm}%
          \[\left\{\hspace*{-1.5ex}\begin{array}{c}T\\T\\\vdots\\T \end{array}\right.\]%
          \end{minipage}};\node[below=.1cm of T1](t1e){};\node[below=.8cm of T1](T1e) {$\varepsilon$};\path (t1e) edge (T1e);
      \node[below left=-.9cm and 0cm of T1,rotate=90]{$k_1+1$ times};
      \node[below right=0cm and -.1cm of C](C1){$C$};
      \node[below left=.5cm and .2cm of C1](T2){\begin{minipage}{1cm}%
          \[\left\{\hspace*{-1.5ex}\begin{array}{c}T\\T\\\vdots\\T \end{array}\right.\]%
          \end{minipage}};\node[below=.1cm of T2](t2e){};\node[below=.8cm of T2](T2e) {$\varepsilon$};\path (t2e) edge (T2e);
      \node[below left=-.9cm and 0cm of T2,rotate=90]{$k_2+1$ times};
      \node[below right=.2cm and 0cm of C1](Cn){$C$};
      \node[below left=.5cm and -.2cm of Cn](Tn){\begin{minipage}{1cm}%
          \[\left\{\hspace*{-1.5ex}\begin{array}{c}T\\T\\\vdots\\T \end{array}\right.\]%
          \end{minipage}};\node[below=.1cm of Tn](tne){};\node[below=.8cm of Tn](Tne) {$\varepsilon$};\path (tne) edge (Tne);
      \node[below right=of Cn](Cnp){$C$};
      \node[below=of Cnp](Tnp){$T$};
      \node[below=of Tnp](Tnpe){$\varepsilon$};
      \node[below left=0cm and 1.6cm of C0](T0){$T$};
      \node[below=of T0](T0e){$\varepsilon$};
      \node[below left=-.9cm and 0cm of Tn,rotate=90]{$k_n+1$ times};
      \path (C) edge (.1,-2.1)
      (C) edge (C1)
      (C0) edge (C) 
      (C0) edge (T0.north)
      (C1) edge (1.15,-2.6)
      (C1) edge[dotted] (Cn)
      (Cn) edge (1.9,-3.4)
      (Cn) edge (Cnp)
      (Cnp) edge (Tnp)
      (Tnp) edge (Tnpe)
      (T0) edge (T0e);
    \end{tikzpicture}}
  \caption{\label{fig-game}A turn of the 2-players corridor game and
    its tree encoding.}
\end{figure}
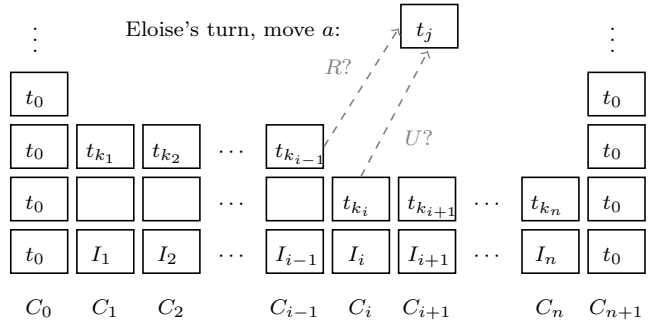
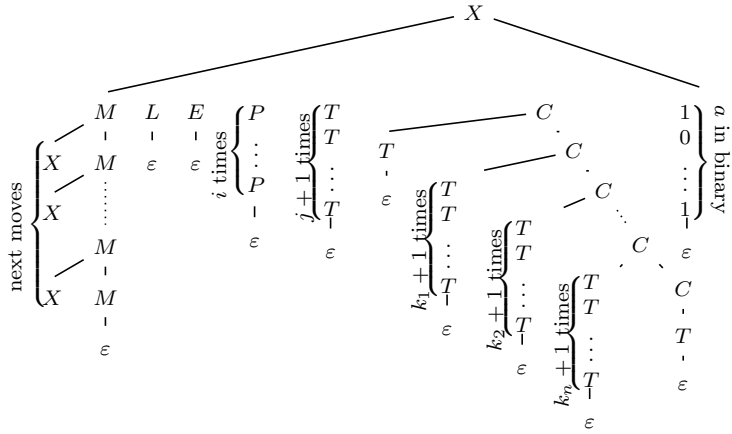

A game turn is encoded locally by an $X$-labeled node and its
immediate children, with the next reachable configurations reachable
through a path of $M$-labeled nodes.  More precisely, each $X$ node
has the following children (see \autoref{fig-game-tree}):
\begin{itemize}
\item a node labeled either $W$ or $L$, stating that the configuration
  is winning or not for Eloise,
\item a node labeled either $E$ or $A$, stating whether it is Eloise's
  or Abelard's turn to move,
\item a chain of $i$ $P$-labeled nodes, stating that the current
  playing column is $C_i$,
\item a chain of $j+1$ $T$-labeled nodes, stating that the chosen
  tile at this turn is $t_j$,
\item a comb-shaped subtree of $n+2$ nodes labeled by $C$, describing the
  contents of the top layer in the corridor (which in general spans
  two rows), with a strand of $k+1$ $T$-labeled nodes telling for each
  column that tile $t_k$ is on top,
\item a chain of $m$ nodes labeled $0$ or $1$, encoding in binary the
  number of the moves made so far; this chain does not need to be
  longer than $m\eqdef\lceil\log(2n^{s+3})\rceil$---or some move would
  have been repeated.
\end{itemize}

\paragraph{The Grammar}  We \emph{fix} $w\eqdef\varepsilon$ and
$G\eqdef\tup{N,\emptyset,P,X}$ over the nonterminal alphabet
$N=\{X,M,W,L,E,A,P,T,C,0,1\}$ with productions (with a slightly
extended syntax with alternatives built-in the productions right-hand
sides):
\begin{align*}
  X&\to M\,(W\mid L)\,(E\mid A)\,P\,T\,C\,(0\mid 1)&
  M&\to X\,M\mid\varepsilon&
  C&\to T\mid T\,C\\
  P&\to P\mid\varepsilon&
  T&\to T\mid\varepsilon&
  W&\to\varepsilon\\
  L&\to\varepsilon&
  E&\to\varepsilon&
  A&\to\varepsilon\\
  0&\to 0\mid 1\mid \varepsilon&
  1&\to 0\mid 1\mid \varepsilon
\end{align*}
Observe that all the tree encodings of strategies are generated by
$G$, but that not all the trees of $G$ encode a strategy: for
instance, the number of $C$'s might be different from $n+2$, or the
described tiling might not respect the placement constraints, etc.  The
formula will check these conditions.

\paragraph{Game Structure and Mechanics}
We use the non-terminal labels and $\varepsilon$ as atomic
propositions in our \PDL\ formula.  Because there are at most $s+2$
choices of tiles at each turn, we can define the path
\begin{equation*}
  \#{move}\eqdef\sum_{i=1}^{s+2}(\downarrow;M?)^{i};\downarrow;X?
\end{equation*}
that relates two successive configurations.  Further define the
following formul\ae\ for $0\leq i\leq n+1$, $0\leq j\leq s+1$, $1\leq
a\leq m$ and $b$ in $\{0,1\}$:
\begin{align*}
  \#p(i)&\eqdef\tup{(\downarrow;P?)^i;\downarrow}\varepsilon&
  \#t(j)&\eqdef\tup{(\downarrow;T?)^{j+1};\downarrow}\varepsilon\\
  \#q(a,b)&\eqdef\tup{(\downarrow;(0+1)?)^{a}}b&
  \#c(i,j)&\eqdef\tup{(\downarrow;
    C?)^{i+1};(\downarrow;T?)^{j+1};\downarrow}\varepsilon
\end{align*}
In an $X$ node, $\#p(i)$ holds if the current column is $C_i$,
$\#t(j)$ if the chosen tile is $T_j$, $\#c(i,k_i)$ if tile $T_{k_i}$
is on top of column $C_i$, and $\#q(m,b)$ if the binary encoding of
the current move number has bit $a$ set to $b$.

We ensure some preliminary structure on the game tree: At the start of
the play, the current player must be Eloise, and the referee should have placed the initial
tiles in the first row.  The counter must be initialized to zero.
\begin{align*}
  \varphi_1 &\eqdef X\wedge
  (\tup{\downarrow}E)\wedge\#p(1)\wedge\bigwedge_{i=1}^n\#c(i,I_i)\wedge\bigwedge_{a=1}^m\#q(a,0)\;.
  \intertext{In every state, the tiles in columns $0$ and $n+1$
    should be the white tile $t_0$.}
  \varphi_2 &\eqdef
  [\downarrow^\ast]X\imply\#c(0,0)\wedge\#c(n+1,0)\;.
\intertext{In every state with current column $i$, the next move
  should be at position $(i\mathbin{\mathrm{mod}}n)+1$.}
  \varphi_3 &\eqdef\bigwedge_{i=1}^n [\downarrow^*] (X \wedge \#p(i))
  \imply [\#{move}] \#p((i\mathbin{\mathrm{mod}}n)+1)\;.
  \intertext{The columns are updated with the correct tiles.}
  \varphi_4 &\eqdef\bigwedge_{i=1}^n\bigwedge_{j=0}^{s+1}[\downarrow^*](X\wedge\#p(i)\wedge\#t(j))\imply[\#{move}]\#c(i,j)\\
         &\wedge\bigwedge_{i\neq j=1}^n\bigwedge_{k=0}^{s+1}(X\wedge\#p(i)\wedge\#c(j,k))\imply[\#{move}]\#c(j,k)\;.
\intertext{Players alternate.}
  \varphi_5&\eqdef[\downarrow^\ast](X\wedge\tup{\downarrow}E\equiv[\#{move}]\tup{\downarrow}A)\;.
  \intertext{The chosen tiles verify the adjacency constraints: define
  for this the Boolean:}
  \#{adj}(i,j,k,\ell)&\eqdef t_\ell\mathbin{U}t_j\wedge(i>0\imply t_k\mathbin{R}t_j)\wedge
  (i=0\imply t_0\mathbin{R}t_j)\wedge(i=n\imply t_j\mathbin{R}t_0)\\
  \varphi_6&\eqdef\bigwedge_{i=1}^n\bigwedge_{j,k,\ell=0}^{s+1}[\downarrow^\ast](X\wedge\#p(i)\wedge\#t(j)\wedge\#c(i-1,k)\wedge\#c(i,\ell))\imply
  \#{adj}(i,j,k,\ell)\;.
\intertext{The counter is incremented.}
      \varphi_7
      &\eqdef\bigwedge_{d=1}^m\bigwedge_{a=1}^{d-1}\bigwedge_{b\in\{0,1\}}[\downarrow^\ast]\big(X
        \wedge \#q(a,b)\wedge \#q(d,0)
        \wedge\bigwedge_{e=d+1}^m\#q(e,1)\big)\\
        &\hspace{9em}\imply[\#{move}]\big(\#q(a,b)\wedge\#q(d,1)\wedge\bigwedge_{e=d+1}^m\#q(e,0)\big)\,.
\end{align*}

\paragraph{Winning Strategy}  The previous formul\ae\ were making
sure that the tree would be a proper game tree.  We want now to check
that it describes a winning strategy for Eloise:  We should check that
all the possible moves of Abelard are tested:
\begin{align*}
  \varphi_8 &\eqdef\bigwedge_{i=1}^n\bigwedge_{j,k,\ell =
    0}^{s+1}[\downarrow^\ast](X\wedge\#p(i)\wedge\tup{\downarrow}E\wedge\#t(k)\wedge\#c((i\mathbin{\mathrm{mod}}n)+1,\ell)\wedge\#{adj}(i,j,k,\ell))\imply\tup{\#{move}}\#t(j)\;.
\intertext{Finally, the winning condition should be met:}
  \varphi_9 &\eqdef (\tup{\downarrow}W) \wedge
  [\downarrow^\ast](X\wedge\tup{\downarrow}W)\imply\big(\#c(1,s+1)\tag{the
    game is immediately winning}\\
  &\phantom{ (\tup{\downarrow}W) \wedge
  [\downarrow^\ast](X\wedge\tup{\downarrow}W)\imply\big(}\vee((\tup{\downarrow}E)\wedge(\tup{\#{move};\downarrow}W))\tag{Eloise
  can win later}\\
  &\phantom{ (\tup{\downarrow}W) \wedge
  [\downarrow^\ast](X\wedge\tup{\downarrow}W)\imply\big(}\vee(\tup{\downarrow}A)\wedge(\tup{\#{move}}\top)\wedge[\#{move}]\tup{\downarrow}W)\big)\;.\tag{None
  of Abelard's moves can prevent Eloise from winning}
\end{align*}

Finally, our final $\PDLcr[\downarrow]$ formula is
$\varphi\eqdef\bigwedge_{i=1}^9\varphi_i$.
Because $G$ and $w$ are fixed and $\varphi$ can be computed in space
logarithmic in the size of the game instance, we have therefore shown
the general PFMC problem to be \textsc{ExpTime}-hard.

\section{Acyclic and {\boldmath $\varepsilon$}-Free Case: Proposition~\ref{prop-aef}}\label{ax-nphard}
We prove here the lower bound part of Proposition~\ref{prop-aef}: the
PFMC problem is \textsc{NPTime}-hard for acyclic and
$\varepsilon$-free grammars.

\begin{proof} We reduce 3SAT to our problem.

Fix the grammar $G\eqdef\tup{\{S,F,T\},\{a\},P,S}$ with productions:
\begin{align*}
  S&\to S\,F\mid S\,T\mid F\mid T&
  F&\to a&
  T&\to a
\end{align*} and consider an instance $\psi=\bigwedge_{i=1}^mC_i$ of
3SAT where each $C_i$ is a disjunction of literals over $n$ variables
$\{x_1,\dots,x_n\}$.  Define $w\eqdef a^n$.

Any parse tree $t$ of $w$ will have a `comb' shape of length $n$
with $S$-labeled nodes, each giving rise to one of $F$ or $T$ as a
child.  The parse forest is thus in bijection with the set of
valuations of $\{x_1,\dots,x_n\}$: if the value of variable $x_i$ is
$0$, then in our encoding, the $i$th $S$ node has a node with label
$F$ as a child; otherwise, it has a node with label $T$ as a child.

Given such an encoded valuation, our formula $\varphi$ must verify
that each clause is satisfied.  For a clause
$C_i=\ell_{i,1}\vee\ell_{i,2}\vee\ell_{i,3}$ with $\ell_{i,j}=x_{k_j}$
or $\ell_{i,j}=\neg x_{k_j}$, define
$\varphi_i\eqdef\bigvee_{j=1}^3\tup{(S;\downarrow)^{k_j}}\beta_{i,j}$
where $\beta_{i,j}=F$ if $\ell_{i,j}=\neg x_{k_j}$ and $\beta_{i,j}=T$
otherwise.  Finally, let $\varphi\eqdef\bigwedge_{i=1}^m \varphi_i$.
Then $t\models\varphi$ if and only if the corresponding assignment of
the variables is a satisfying assignment.  Because $G$ is fixed and
$w$ and $\varphi$ can be computed in space logarithmic in the size of
the 3SAT instance, this shows the \textsc{NPTime}-hardness of the PFMC
problem in the acyclic $\varepsilon$-free case.  Note that $\varphi$
is in $\PDLcr[\downarrow]$.
\end{proof}

\section{Model-Checking Non-Recursive DTDs: Proposition~\ref{prop-dtd}}\label{ax-dtd}
We present in this section a proof of Proposition~\ref{prop-dtd}: the
satisfiability of a \PDL\ formula $\varphi$ in presence of a
non-recursive DTD $D$ is \textsc{PSpace}-complete.

The lower bound is proved as Proposition~5.1 by \citet{benedikt08},
and we follow their general proof plan from Lemma~7.5 for the upper
bound.  As presented in the main text, the fact that we consider a
non-recursive DTD means that the height of any tree of interest is
bounded by $|N|$ the number of nonterminals of the DTD.  The proof
plan is then to consider XML word encodings of trees, and construct
two \emph{2-way alternating parity word automata} (2APWA) $\?A_D$ and
$\?A_\varphi$ of polynomial size which will respectively recognize the
XML encodings of the trees of $D$ and of the models of $\varphi$ of
height bounded by $|N|$.  Then, by taking the conjunction of the two
automata, we reduce the initial satisfiability problem to a 2APWA
emptiness problem, which is known to be in \textsc{PSpace} by the
results of \citet{serre06}.

We can find a suitable construction for an automaton $\?A_D$ for $D$ as
Claim~7.7 of \citep{benedikt08}, thus we will only present the
construction of $\?A_\varphi$.

\newcommand{\cl}[1]{\tup{/#1}}
\paragraph{XML Encoding}  Define the alphabet
\begin{align*}
  \#{XML}(N)&\eqdef\{\tup{X}, \tup{/X}\mid X \in N\}
  \intertext{%
    and choose a fresh root symbol $\#r$ not in $N$.  We encode our a
    tree $t$ as $\tup{\#r}\#{stream}(t)\cl{\#r}$ where the \emph{XML
    streaming} function is defined inductively on terms by}
  \#{stream}(f(t_1\cdots t_m))&\eqdef
  \tup{f}\#{stream}(t_1)\cdots\#{stream}(t_m)\cl{f}\;.
\end{align*}

\paragraph{2-Way Alternating Parity Word Automata}  A
\emph{positive boolean formula} $f$ in $\mathbb{B}^+(X)$ over a set $X$ of
variables is defined by the syntax
\begin{align*}
  f &::= \top\mid\bot\mid f\wedge f\mid f\vee f\;.
\end{align*}
A subset $X'\subseteq X$ satisfies a formula $f$, written $X'\models
f$, if the formula is satisfied by the valuation $x\mapsto \top$
whenever $x\in X'$ and $x\mapsto\bot$ if $x\in X'\setminus X$.

A \emph{2-way alternating parity word automaton} is a tuple
$\?A=\tup{Q,\Sigma,\delta,q_0,c}$ where $Q$ is a finite set of states,
$\Sigma$ a finite alphabet, $q_0\in Q$ an initial state, $c$ a
coloring from $Q$ to a finite set of priorities
$C\subseteq\mathbb{N}$, and $\delta$ a transition function from
$Q\times\Sigma$ to $\mathbb{B}^+(Q\times\{-1,0,1\})$ that associates
to a current state and current symbol boolean formul\ae\ on pairs
$(q',d)$ of a new state $q'$ and a direction $d$.

A \emph{run} of a 2APWA on a finite word $w=a_1\cdots a_n$ in
$\Sigma^\ast$ is a generally infinite tree with labels in
$Q\times\{1,\dots,n\}$ holding a current state and a current position
in $w$, such that the root is labeled $(q_0,1)$, and every node
labeled $(q,i)$ with has a children set
$\{(q_1,i_1),\dots,(q_m,i_m)\}$ that satisfies $\delta(q,a_i)$.  A run
is \emph{accepting} iff for every branch, the smallest priority $c(q)$
that occurs infinitely often among the nodes $(q,i)$ is even---this
also means in particular that any finite run is accepting---, and $w$
is accepted if there exists some accepting run for it.

\paragraph{Inductive Construction}
We construct
$\?A_\varphi\eqdef\tup{Q_\varphi,\Sigma,\delta_\varphi,q_{0,\varphi},c_\varphi}$
by induction on the subterms of the formula $\varphi$.
We work with the alphabet
$\Sigma\eqdef\mathsf{XML}(N)\uplus\{\tup{\#r},\cl{\#r}\}$ and set
$n\eqdef|N|$---which is the maximum height of any tree of $D$.  The
guiding principles in this construction is that our inductively
constructed automaton will track their height relative to that of
their starting position.  Because we are working on trees of
bounded depth, this can be achieved by considering states that combine
a `control' state with a height in $\{-n,\dots,n\}$. 

Let us start with the base cases for node formul\ae: by convention, our
automata for a node formul\ae\ must check that their starting
positions are labeled by opening tags:
\begin{description}
\item[{\boldmath $\?A_p$}] The automaton checks
  if it starts at an opening node $\tup p$. It immediately goes into either
  an accepting or rejecting state.  Formally, $Q_p\eqdef\{q_{p,0}\}$, the
  coloring $c_p$ maps $q_{p,0}$ to $1$, and
  $\delta(q_{p,0},\tup{p})\eqdef\top$ and $\delta_p(q_{p,0},X)=\bot$ for all
  $X\neq\tup{p}$.
\item[{\boldmath $\?A_\top$}]  The automaton immediately goes into
  an accepting state, unless it is at a closing node or the root node.
  Formally, $Q\eqdef\{q_{\top,0}\}$ and $c_\top$ is defined by
  $c_\top(q_{\top,0})\eqdef 1$; $\delta_\top(q_0,X)$ is defined as $\top$ for
  $X=\tup{p}$ in $\#{XML}(N)$ and as $\bot$ otherwise.
\end{description}\bigskip

The automata $\mathcal{A}_\pi$ for $\pi$ a path formula additionally
carry a distinguished subset $C_\pi\subseteq Q_\pi$ of
\emph{continuation} states, such that 
there is a `partial run' from some initial position with branches
starting from their initial state, which are either infinite but
verifying the parity condition, or are finite but end in a continuation
state in a position related to the initial one through $\sem{\pi}$.
Let us see this at work with the base cases of path formul\ae:
\begin{description}
\item[{\boldmath $\?A_\downarrow$}] The automaton moves right from the 
  initial node while maintaining the depth relative to this initial
  node.  It stops (goes into a dead state) if it reaches a node at
  the same or lesser depth than the initial node.  All the 
  visited nodes with a relative depth of 1 are direct children of the
  initial node, and therefore visited by continuation states.  We set
  where $Q_\downarrow\eqdef\{q_0,q_1,\cdots,q_n\}$ with
  $q_{\downarrow,0}\eqdef q_0$; the coloring $c_\downarrow$ is
  identically $1$ on $Q_\downarrow$, $C_\downarrow=\{q_1\}$, and
  $\delta_\downarrow$ is defined by:
  \begin{align*}
    \delta_\downarrow(q_i, \langle p \rangle) &\eqdef (1,q_{i+1}), &&i<n, p\in N&
    \delta_\downarrow(q_n, \langle p \rangle) &\eqdef\bot, &&p \in N \\
    \delta_\downarrow(q_i, \langle /p \rangle) &\eqdef(1,q_{i-1}), &&i>1, p \in N&
    \delta_\downarrow(q_0, \langle /p \rangle)&\eqdef\delta_\downarrow(q_1, \langle /p \rangle)\eqdef\bot, &&p \in N \\
    \delta_\downarrow(q_i, \langle /\mathsf{r} \rangle) &\eqdef\bot, && i \in \{0,\dots,n\}\;.
  \end{align*}
  \item[{\boldmath $\?A_\rightarrow$}] Similarly to $\?{A}_{\downarrow}$, the
    automaton moves right while maintaining the depth relative to the
    initial node.  It fails if it reaches a node at a lesser depth
    that the initial node.  Otherwise, it finds the next node at a
    same depth as the initial node.  Formally,
    $Q_\rightarrow\eqdef\{q_0,q_1,\dots,q_n,q_f\}$,
    $q_{\rightarrow,0}\eqdef q_0$, the coloring
    $c_\rightarrow$ is identically $1$ on $Q_\downarrow$, there is a
    unique continuation state  $C_\rightarrow\eqdef\{q_f\}$ when
    reaching the right sibling, and $\delta_\rightarrow$ is defined by
  \begin{align*}
    \delta_\rightarrow(q_i, \langle p \rangle) &\eqdef (1,q_{i+1}), &&i<n, p \in N &
    \delta_\rightarrow(q_n, \langle p \rangle) &\eqdef\bot, &&p \in N \\
    \delta_\rightarrow(q_i, \langle /p \rangle) &\eqdef (1,q_{i-1}), &&i>1 , p \in N &
    \delta_\rightarrow(q_1, \langle /p \rangle) &\eqdef (1,q_f), &&p \in N \\
    \delta_\rightarrow(q_0, \langle /p \rangle) &\eqdef\bot, &&p \in N &
    \delta_\rightarrow(q_f, X) &\eqdef\bot, && X\in \Sigma\\
    \delta_\rightarrow(q_i, \langle /\mathsf{r} \rangle) &\eqdef\bot, && i \in \{0,\dots,n\}\;.
  \end{align*}
  \item[{\boldmath $\?A_\uparrow,\?A_\leftarrow$}]\IEEEhspace{2em} We define these
    automata similarly to $\?A_\downarrow$ and $\?A_\rightarrow$.
    Observe however that, because we always finish on opening
    brackets, it is not enough to exchange $-1$ and $1$ in the
    directions of transitions.
\end{description}\bigskip

Next, we consider the induction step for path formul\ae:
\begin{description}
\item[{\boldmath $\?A_{\pi_1;\pi_2}$}]\IEEEhspace{1em}  We combine the automata
    $\mathcal{A}_{\pi_1}$ and $\mathcal{A}_{\pi_2}$. We add
  transitions from the continuation states of $\mathcal{A}_{\pi_1}$ at
  opening nodes to the initial state of $\mathcal{A}_{\pi_2}$.
  Formally, $Q_{\pi_1;\pi_2}\eqdef Q_{\pi_1}\uplus Q_{\pi_2}$,
    $q_{\pi_1;\pi_2}\eqdef q_{\pi_1,0}$, $c_{\pi_1;\pi_2}$ preserves
  the priorities of $c_{\pi_1}$ and $c_{\pi_2}$,
  $C_{\pi_1;\pi_2}\eqdef C_{\pi_2}$ and $\delta_{\pi_1;\pi_2}$ is defined by
    \begin{align*}
     \delta_{\pi_1;\pi_2}(q_1,X) &\eqdef \delta_{\pi_1}(q_1,X),
     && q_1 \in Q_{\pi_1} \setminus C_{\pi_1}\\
     \delta_{\pi_1;\pi_2}(q_1,\langle p \rangle) &\eqdef\delta_{\pi_1}(q_1,\langle p \rangle) \vee (0,q_{\pi_2,0}), && q_1 \in C_{\pi_1}\\
     \delta_{\pi_1;\pi_2}(q_1,\langle /p \rangle) &\eqdef\delta_{\pi_1}(q_1,\langle /p \rangle), && q_1 \in C_{\pi_1}\\
     \delta_{\pi_1;\pi_2}(q_2,X) &\eqdef \delta_{\pi_2}(q_2,X), && q_2 \in Q_{\pi_2}
    \end{align*}
    for $X$ in $\Sigma$ and $p$ in $N$.
  \item[{\boldmath $\?A_{\pi_1+\pi_2}$}]\IEEEhspace{1.1em} This is a straightforward
    union: We define $Q_{\pi_1+\pi_2}\eqdef Q_{\pi_1}\uplus Q_{\pi_2}\uplus\{q_{\pi_1+\pi_2,0}\}$, $C_{\pi_1+\pi_2}\eqdef C_{\pi_1} \cup C_{\pi_2}$,
    $c_{\pi_1+\pi_2}\eqdef c_{\pi_1}\cup c_{\pi_2}\cup\{(q_{\pi_1+\pi_2,0},1)\}$, $\delta_{\pi_1+\pi_2}\eqdef\delta_{\pi_1} \cup \delta_{\pi_2}\cup 
    \{(q_{\pi_1+\pi_2,0}, X,(0,q_{\pi_1,0}) \vee (0,q_{\pi_2,0})) \mid X \in \Sigma\}$.
  \item[{\boldmath $\?A_{\pi^\ast}$}] This case is similar to that of
    $\?A_{\pi_1;\pi_2}$; we add transitions from the critical states
    of $\mathcal{A}_\pi$ to its own initial state.  Define
    $Q_{\pi^\ast}\eqdef Q_\pi\uplus\{q_{\pi^\ast,0}\}$,
    $C_{\pi^\ast}\eqdef\{q_{\pi^\ast,0}\}$, 
    $c_{\pi^\ast}\eqdef c_\pi\cup\{(q_{\pi^\ast,0},1)\}$, and 
    \begin{align*}
      \delta_{\pi^\ast}(q_{\pi^\ast,0},X) &\eqdef (0,q_{\pi,0}), && &
      \delta_{\pi^\ast}(q,\langle /p \rangle) &\eqdef \delta_\pi(q,X), && q \in C_\pi\\
      \delta_{\pi^\ast}(q,X) &\eqdef \delta_\pi(q,X), && q \in Q_\pi \setminus C_\pi &
      \delta_{\pi^\ast}(q,\langle p \rangle) &\eqdef \delta_\pi(q,X) \vee (0,q_{\pi^\ast,0}), && q\in C_\pi
    \end{align*}
    for $X$ in $\Sigma$ and $p$ in $N$.  Because we assigned an odd
    priority to $q_{\pi^\ast,0}$, the automaton cannot loop
    indefinitely in $q_{\pi^\ast,0}$ and must eventually continue.
  \item[{\boldmath $\?A_{\psi?}$}] Define $Q_{\psi?}\eqdef
    Q_\psi\uplus\{q_{\psi?,0},q_f\}$, $C_{\psi?}\eqdef\{q_f\}$,
    $c_{\psi?}$ extends $c_\psi$ with $c_{\psi?}(q_{\psi?,0})=c_{\psi?}(q_f)=1$, and
    \begin{align*}
      \delta_{\psi?}(q_{\psi?,0},X) &\eqdef (0,q_f) \wedge (0,q_{\psi,0}), &
      \delta_{\psi?}(q_f,X) &\eqdef\bot, &
      \delta_{\psi?}(q,X) &\eqdef \delta_\psi(q,X),
    \end{align*}
    for $X$ in $\Sigma$ and $q$ in $Q_\psi$.
\end{description}\bigskip

Finally, we consider the induction step for node formul\ae:
\begin{description}
\item[{\boldmath $\?A_{\tup\pi\psi}$}]\IEEEhspace{.2em} We construct the automaton by
  joining the critical states of $\mathcal{A}_\pi$ with the initial
  state of $\mathcal{A}_\psi$. Define $Q_{\tup\pi\psi}\eqdef Q_\pi\uplus Q_\psi$, $q_{\tup\pi\psi,0}\eqdef q_{\pi,0}$, $c_{\tup\pi\psi}\eqdef c_\pi\cup c_\psi$, and $\delta_{\tup\pi\psi}$ by
    \begin{align*}
     \delta_{\tup\pi\psi}(p,X)&\eqdef\delta_\pi(p,X), && p\in Q_\pi\setminus C_\pi&
     \delta_{\tup\pi\psi}(q,X)&\eqdef\delta_\psi(q,X), && q \in
     Q_\psi\\
     \delta_{\tup\pi\psi}(p,X)&\eqdef\delta_\pi(p,X)\vee(0,q_{\psi,0}), && p\in C_\pi
    \end{align*}
    for $X$ in $\Sigma$.
  \item[{\boldmath $\?A_{\psi_1\wedge\psi_2}$}]\IEEEhspace{1.1em} We do a simple
    conjunction of the automata $\?A_{\psi_1}$ and $\?A_{\psi_2}$: define $Q_{\psi_1\wedge\psi_2}\eqdef Q_{\psi_1}\uplus Q_{\psi_2}\uplus\{q_{\psi_1\wedge\psi_2,0}\}$, 
    $c_{\psi_1\wedge\psi_2}\eqdef c_{\psi_1} \cup c_{\psi_2} \cup \{(q_{\psi_1\wedge\psi_2,0},1)\}$, $\delta_{\psi_1\wedge\psi_2}\eqdef \delta_{\psi_1} \cup \delta_{\psi_2} \cup 
    \{(q_{\psi_1\wedge\psi_2,0},X, (0,q_{\psi_1,0}) \wedge (0,q_{\psi_2,0})) \mid X \in \Sigma\}$.

  \item[{\boldmath $\?A_{\neg\psi}$}] We essentially construct the dual
    $\#{dual}(\?A_\psi)$ of $\mathcal{A}_\psi$: the latter accepts the
    complement of the language accepted by
    $\mathcal{A}_\psi$. However, we need to ensure that only opening
    nodes $\langle p \rangle$ are accepted, thus intersect with the
    automaton $\mathcal{A}_\top$ that only accepts opening nodes.
    Formally,
    $\#{dual}(\?A_\psi)=\tup{Q_\psi,\Sigma,\delta_\psi,q_{\psi,0},c_{\neg\psi}}$
    where $\delta_{\neg\psi}(q,X)\eqdef\#{dual}(\delta_\psi(q,X))$ and
    $c_{\neg\psi}(q) \eqdef c_\psi(q) + 1$ for all $q \in Q$ and $X
    \in \Sigma$.  Here, $\#{dual}$ is a function from
    $\mathbb{B}^+(Q\times\{-1,0,1\})$ to itself that applies the usual
    DeMorgan's law.  It is easy to check that $\#{dual}(\mathcal{A}_\psi)$
    accepts the complement of $L(\mathcal{A}_\psi)$.
\end{description}

\section {{\boldmath $\varepsilon$}-Free Case}\label{ax-efhard}
We prove in this section %
Proposition~\ref{prop-efph}, thus showing that the PFMC problem is
\textsc{PSpace}-hard in this case.

\begin{proof}
  We reduce from the membership problem of linear bounded automata (LBA).
  Suppose we are given an LBA $M=\tup{Q,\Gamma,\Sigma,\delta,q_1,F}$
  with state set $Q$, tape alphabet $\Gamma$, input alphabet
  $\Sigma\subseteq\Gamma$, transition relation $\delta\subseteq
  Q\times\Gamma\times Q\times\Gamma\times\{-1,0,1\}$, initial state
  $q_1\in Q$, and set of final states $F\subseteq Q$.
  Let $Q=\{q_1,\dots,q_\ell\}$; we assume that
  $\Gamma=\{a_1,a_2,\dots,a_m\}$ contains two endmarkers $a_1 =
  \triangleleft$ and $a_2 = \triangleright$ that surround the input
  and are never erased nor crossed during the run of the machine.

  We are also given a string $x=b_1b_2\cdots b_n$ with each
  $b_i\in\Sigma$; $b_1 = \triangleleft$ and $b_n = \triangleright$. We
  have to decide whether $x$ is accepted by $M$.  We are going to
  construct a word $w$, a CFG $G$, and a $\PDLcr[\downarrow]$ formula
  $\varphi$, s.t.\ the PFMC problem has a solution for
  $\tup{w,G,\varphi}$ iff $M$ accepts $x$.

  \paragraph{Encoding as Linear Trees}
  A configuration of $M$ is a sequence of length $n$ of form
  $\triangleleft \gamma q\gamma'\triangleright$ where $q$ is the
  current state in $Q$, $\triangleleft\gamma\gamma'\triangleright$ is
  the current tape contents, and $|\triangleleft\gamma|=h$ indicates
  that the head is currently on the last symbol of
  $\triangleleft\gamma$, i.e.\ the $h$th symbol of the tape.

  We encode such a configuration by a contiguous sequence $\alpha$ of nodes
  as follows:
  \begin{itemize}
  \item The first node is $S$ and it is followed by a sequence of $n$
    nodes, among which one is labeled $H$ and the others $\bar{H}$; the
    position of $H$ in this sequence denotes the position of the head
    in the configuration of $M$.
  \item This sequence is followed by a sequence of $\ell$ nodes, one
    labeled $C$ and the others $\bar{C}$, which together describe the current
    state as $q_k$ if the occurrence of $C$ is the $k$th symbol in the
    sequence.
  \item Then we encode the tape contents as $n$ successive sequences
    each of length $m$ of nodes, with each time one labeled $A$ and
    the others $\bar{A}$. The
    $i$th such sequence encodes the contents of the $i$th cell of the
    tape of $M$ with $A$ occurring at the $j$th position indicating
    that this cell contains $a_j$.
  \end{itemize}
  Thus $\alpha$ is of length $1+n+\ell+n m$.

  \paragraph{String and Grammar}
  We \emph{fix} $w\eqdef a$ and $G\eqdef\tup{\{S,H,\bar{H},A,\bar{A},C,\bar{C}\},\{a\},P,S}$
  with productions
  \begin{align*}
    S&\to H\mid \bar{H}\\
    H&\to H\mid\bar{H}\mid C\mid\bar{C}&
    C&\to C\mid\bar{C}\mid A\mid\bar{A}&
    A&\to A\mid\bar{A}\mid S\mid a\\
    \bar{H}&\to H\mid\bar{H}\mid C\mid\bar{C}&
    \bar{C}&\to C\mid\bar{C}\mid A\mid\bar{A}&
    \bar{A}&\to A\mid\bar{A}\mid S\mid a\;.
  \end{align*}
  Therefore, the trees in the parse forest $L_{G,w}$ are essentially
  sequences over $N^\ast\cdot\{a\}$.  Clearly, all the encodings of finite
  runs of any LBA $M$ will be in this set; it will be the formula's
  task to look for an accepting run of our particular $M$ on $x$ among
  all these trees.

  \paragraph{The Formula {\boldmath $\varphi_{M,x}$}}
  Let us turn to the definition of our $\PDLcr[\downarrow]$ formula.  We
  start by defining low-level formul\ae\ useful for testing the
  properties of the current configuration: assume we are on an
  $S$-labeled node:
  \begin{align*}
    \#h(h)&\eqdef(\tup{\downarrow^h}H)\wedge\bigwedge_{i\in\{1,\dots,n\}\setminus\{h\}}\tup{\downarrow^j}\bar{H}
    \intertext{tests whether the head is at position $h$.  In the
      same way,}
    \#q(k)&\eqdef\tup{\downarrow^n}\big((\tup{\downarrow^k}C)\wedge\bigwedge_{k'\in\{1,\dots,\ell\}\setminus\{k\}}\tup{\downarrow^{k'}}\bar{C}\big)
    \intertext{then tests whether the current state is $q_k$, and}
    \#p(i,j)&\eqdef\tup{\downarrow^{n+\ell+i
        m}}\big((\tup{\downarrow^j}A)\wedge\bigwedge_{j'\in\{1,\dots,m\}\setminus\{j\}}\tup{\downarrow^{j'}}\bar{A}\big)
    \intertext{tests whether the $i$ position on the tape is symbol
      $a_j$.  Finally, we can go to the next configuration by the path}
    \#{next}&\eqdef\downarrow^{n+\ell+n m}\;.
    \intertext{We can now check that a parse tree of $L_{G,w}$ is
      really the encoding of an accepting run of $M$ on $x$.  First,
      at each $S$ node, we should find a full configuration:}
    \varphi_{\mathit{conf}}&\eqdef
           [\downarrow^\ast]S\imply\bigvee_{h=1}^n\#h(h)\wedge\bigvee_{k=1}^\ell\#q(k)\wedge\bigwedge_{i=1}^n\bigvee_{j=1}^m\#p(i,j)\wedge\tup{\#{next}}a\vee
           S\;.
    \intertext{The initial configuration should have its head on the
      initial position $1$, be in the initial state $q_1$, and have $x=b_1\cdots
      b_n$ as tape contents:}
    \varphi_{\mathit{init}}&\eqdef
    S\wedge\#h(1)\wedge\#q(1)\wedge\bigwedge_{i=1}^n\#p(i,b_i)\;.
    \intertext{The leaf of the tree should be reached in a final
      configuration:}
    \varphi_{\mathit{final}}&\eqdef[\downarrow^\ast](S\wedge\tup{\#{next}}a)\imply\bigvee_{q_k\in
      F}\#q(k)\;.
    \intertext{Successive configurations should respect the
      transition relation:}
    \varphi_{\mathit{trans}}&\eqdef[\downarrow^\ast](S\wedge\neg\tup{\#{next}}a)\imply\bigvee_{h=1}^n\bigvee_{k=1}^\ell\bigvee_{c=1}^m\bigvee_{(q_k,a_c,q_{k'},a_{c'},d)\in\delta}\big(\#h(h)\wedge\#q(k)\wedge\#p(h,c)\\&\wedge\big(\bigwedge_{h\neq
      i=1}^n\bigvee_{j=1}^m\#p(i,j)\wedge\tup{\#{next}}\#p(i,j)\big)\wedge\tup{\#{next}}(\#h(h+d)\wedge\#q(k')\wedge\#p(h,c'))\big)\,.
    \intertext{We finally define our $\PDLcr[\downarrow]$ formula as
      the conjunction of the previous formul\ae:}
    \varphi_{M,x}&\eqdef\varphi_{\mathit{conf}}\wedge\varphi_{\mathit{init}}\wedge\varphi_{\mathit{final}}\wedge\varphi_{\mathit{trans}}\;.
\end{align*}
To conclude, we observe that a tree in $L_{G,w}$ is a model of
$\varphi_{M,x}$ iff there is an accepting run of $M$ on $x$.  As $G$
and $w$ are fixed and $\varphi_{M,x}$ can be computed in space
logarithmic in the size of $\tup{M,x}$, this proves the
\textsc{PSpace}-hardness of the PFMC problem in the $\varepsilon$-free
case.
\end{proof}

\end{document}